%% file: cloud_offloading.tex
\documentclass[10pt, conference, a4paper]{IEEEtran}

\usepackage[T1]{fontenc}
\usepackage{amsmath}
\usepackage{amssymb}
\usepackage{amsthm}
\usepackage{bbm}
\usepackage{bm}
\usepackage[normalem]{ulem}
\usepackage{enumerate}
\usepackage{comment}
\usepackage{color}
\usepackage{tikz}
\usetikzlibrary{arrows,shapes,decorations.pathmorphing,backgrounds,positioning,fit,matrix}
\usepackage{xcolor}
\usepackage[margin=20mm]{geometry}
\usetikzlibrary{calc}
\usetikzlibrary{matrix}

\usepackage{graphicx}

\usepackage{todonotes}

\usepackage[utf8]{inputenc}
\usepackage{algorithm}
\usepackage{algpseudocode}
\usepackage{paralist}
\usepackage{tikz}
\usetikzlibrary{matrix}

\def\varplayersset{{\mathcal{K}}}
\def\varplayerssetdim{K}
\def\varAPs{{\mathcal{I}}}

\def\varAPsdim{I}
\def\varcloudcapability{F^c}
\def\varusercapability{F^0}
\def\vardecisionsset{\mathfrak{D}}
\def\vardecision{d}
\def\vardecisionsvector{\textbf{d}}
\def\varAP{i}
\def\varoAP{j}
\def\varpower{P}
\def\varchannelgain{H}
\def\varinterferencepower{\omega}
\def\varchannelbandwidth{B}
\def\varuplinkrate{R}
\def\vardatasize{D}
\def\varlocalcomputingtime{T^0}
\def\varCPUcyclesnumber{L}
\def\varlocalcomputingenergy{E^0}
\def\varenergyconstant{v}
\def\varlocalcost{C^0}
\def\vartimeweight{\gamma^T}
\def\varenergyweight{\gamma^E}
\def\varcloudtransmissiontime{T^{c,off}}
\def\varcloudcomputingtime{T^{c,exe}}
\def\varcloudtransmissionenergy{E^c}
\def\varcloudcost{C^c}
\def\varindicatorfunction{I}
\def\vartreshold{M}
\def\varcostfunction{C}
\def\varoplayer{{k^{\prime}}}
\def\varplayer{k}
\def\varnumberofusers{n}
\def\vargamestage{t}
\def\varthegame{\Gamma}
\def\varoffloaders{O}

\def\vardeviatorsOtoL{D_{O\rightarrow L}}
\def\vardeviatorsLtoO{D_{L\rightarrow O}}
\def\vardeviatorsOtoO{D_{O\rightarrow O}}
\def\varreluctanceratio{\rho}
\def\varcostvector{\gamma}

\def\argmax{\mathop{\rm arg\,max}}
\def\argmin{\mathop{\rm arg\,min}}

\algrenewcommand\algorithmicindent{0.7em}

\newtheorem{theorem}{Theorem}
\newtheorem{lemma}{Lemma}
\newtheorem{corollary}{Corollary}

\newtheorem{proposition}[theorem]{Proposition}
\newtheorem{example}{Example}
\newtheorem*{example*}{Example}
\theoremstyle{definition}
\newtheorem{definition}{Definition}
\newcommand*\ruleline[1]{\par\noindent\raisebox{.8ex}{\makebox[\linewidth]{\hrulefill\hspace{1ex}\raisebox{-.8ex}{#1}
\hspace{1ex}\hrulefill}}}
\title{Selfish Computation Offloading for Mobile Cloud Computing in Dense Wireless Networks}
 \author{Sla\dj ana Jo\v{s}ilo and Gy\"orgy D\'an\\
 ACCESS Linnaeus Center, School of Electrical Engineering\\
 KTH, Royal Institute of Technology, Stockholm, Sweden
 E-mail: \{josilo, gyuri\}@kth.se 
 }
\algdef{SE}[DOWHILE]{Do}{doWhile}{\algorithmicdo}[1]{\algorithmicwhile\ #1}%

\begin{document}

\setlength{\abovedisplayskip}{0.18cm}
\setlength{\belowdisplayskip}{0.18cm}

\maketitle
\thispagestyle{plain}
\pagestyle{plain}

\begin{abstract}
\input{abstract.tex}

\end{abstract}

\input{introduction.tex}

\input{model.tex}

\input{equal_elastic.tex}

\input{equal_non_elastic.tex}

\input{PoA.tex}

\input{numerical.tex}

\input{related.tex}

\input{conclusion.tex}

\bibliographystyle{IEEEtran}
\bibliography{IEEEabrv,refs-draft}  
\end{document}

%% file: abstract.tex
Offloading computation to a mobile cloud is a promising solution to augment the computation capabilities of mobile devices. In this paper we consider selfish 
mobile devices in a dense wireless network, in which individual mobile devices can offload computations via multiple access points (APs) to a mobile cloud so 
as to minimize their computation costs, and we provide a game theoretical analysis of the problem. We show that in the case of an elastic cloud, all improvement
paths are finite, and thus a pure strategy Nash equilibrium exists and can be computed easily. In the case of a non-elastic cloud we show that improvement paths may cycle, yet we show that 
a pure Nash equilibrium exists and we provide an efficient algorithm for computing one. Furthermore, we provide an upper bound on the price of anarchy (PoA) of the game. 
We use simulations to evaluate the time complexity of computing Nash equilibria and to provide insights into the PoA under realistic scenarios. Our results show that
the equilibrium cost may be close to optimal, and the cost difference is due to too many mobile users offloading simultaneously.

%% file: introduction.tex
\section{Introduction}
Mobile handsets are increasingly used for various computationally intensive applications,
including augmented reality, natural language processing, face, gesture and object recognition,
and various forms of user profiling for recommendations~\cite{Hakkarainen2008Ismar,Liu2009Percom}. Executing such computationally intensive applications
on mobile handsets may result in slow response times, and can also be detrimental to battery life, which may limit user acceptance. 

Mobile cloud computing has emerged as a promising solution to serve the computational
needs of these computationally intenstive applications, while potentially relieving the battery of the mobile handsets~\cite{Cuervo2010MobiSys,Wen12Infocom}.
In the case of mobile cloud computing the mobile devices offload the computations via a wireless network to a cloud infrastructure,
where the computations are performed, and the result is sent back to the mobile handset. While computation offloading
to general purpose cloud infrastructures, such as Amazon EC2, may not be able to provide sufficiently low response times
for many applications, emerging mobile edge computing resources may provide sufficient computational power close
to the network edge to meet all application requirements~\cite{ETSI2015mec}.

Computation offloading to a mobile edge cloud can significantly increase the computational capability
of individual mobile handsets, but the response times may suffer when many handsets attempt to offload
computations to the cloud simultaneously, on the one hand due to the competition for possibly constrained
edge cloud resources, on the other hand due to contention in the wireless access~\cite{Barbera2013Infocom,Chen2015tpds}. 
The problem is even more
complex in the case of a dense deployment of access points, e.g., cellular femtocells or WiFi access points,
when each mobile user can choose among several access points to connect to. Good system performance
in this case requires the coordination of the offloading choices of the indvidual mobile handsets, while respecting
their individual performance objectives, both in terms of response time and energy consumption.

In this paper we consider the problem of resource allocation for computation offloading by self-interested mobile users
to a mobile cloud. The objective of each mobile user is to minimize a
linear combination of its response time and its energy consumption for performing a computational task, by choosing whether or not
to offload through one of many access points. Clearly, the choice of a mobile user affects the cost of other mobile users. If too many
mobile users choose offloading through a particular access point then they will achieve low transmission rate. A low transmission rate would
lead to high data transmission time and a corresponding high energy consumption. 
In order to capture the interactions between the choices of the mobile users, in this paper we formulate the computation offloading problem as a non-cooperative game,
and address the existence of self-enforcing resource allocations, i.e., equilibrium allocations, and their computation. 

Our contibutions in this paper are threefold. First, we show that if the cloud computing resources scale with the number of mobile users then equilibrium allocations
always exist, and we provide a simple algorithm for computing an equilibrium. Second, we show that if the cloud computing resources do not scale with the number
of mobile users then the same algorithm cannot be used for computing an equilibrium as it may cycle infinitely, but we prove that equilibria exist, and we
provide an algorithm with quadratic complexity in the number of mobile users for computing an equilibrium.
Finally, we provide a bound on the price of anarchy for both models of cloud resources. We provide numerical
results based on extensive simulations to illustrate the computational efficiency of the algorithms and to evaluate the price of anarchy for scenarios of practical interest.

The rest of the paper is organized as follows. We present the system model in Section~\ref{sec:model}. We prove equilibrium existence and computability
results for the elastic cloud and non-elastic cloud in Sections~\ref{sec::elastic} and~\ref{sec::non-elastic}, respectively. We provide a bound
on the price of anarchy in Section~\ref{sec::poa} and present numerical
results in Section~\ref{sec::numerical}. Section~\ref{sec::related} discusses related work and Section~\ref{sec::conclusion} concludes the paper.

%% file: model.tex
\section{System Model and Problem Formulation}
\label{sec:model}
We consider a mobile cloud computing system that serves a set $\varplayersset\!\! = \!\!\{1,2,...,\varplayerssetdim\}$ of colocated mobile users (MU). 
Each MU has a computationally intensive task to perform, and can decide whether to perform the task  locally or to offload the computation to a cloud server.
The computational task is characterized by the size $\vardatasize_\varplayer$ of the input data (e.g., in bytes),
and by the number $\varCPUcyclesnumber_\varplayer$ of the CPU cycles required to perform the computation.
To enable a meaningful analysis, we make the common assumption that the set of MUs does not change during computation offloading, i.e., in
the order of seconds~\cite{Wen12Infocom,Yang:2013:FPE:2479942.2479946,Sardellitti2015tsipn,Iosifidis2013WiOpt}. 
\subsection{Communication model}
If the MU decides to offload the computation to the cloud server, it has to transmit $\vardatasize_\varplayer$ amount of data pertaining
to its task to the cloud server through one of a set of access points (APs)
denoted by $\varAPs \!\!=\!\!  \{1,2,...,\varAPsdim\}$.
Thus, together with local computing MU $\varplayer$ can choose an action from the set  $\vardecisionsset_\varplayer \!\!=\!\! \{0,1,2,...,\varAPsdim\}$,
where $0$ corresponds to local computing, i.e., no offloading.
We denote by $\vardecision_\varplayer\!\in\!\vardecisionsset_\varplayer$ the decision of MU $\varplayer$, and refer to it as her strategy.
We refer to the collection $\vardecisionsvector\!=\!(\vardecision_\varplayer)_{\varplayer\in\varplayersset}$ as a strategy profile, and we denote
by $\vardecisionsset\!=\!\times_{\varplayer \in \varplayersset}\vardecisionsset_\varplayer$ the set of all feasible strategy profiles.

We denote by $\varchannelbandwidth_\varAP$ the bandwidth of AP $\varAP$, and
for a strategy profile $\vardecisionsvector$ we denote by $\varnumberofusers_\varAP(\vardecisionsvector)$ the number of MUs that use AP $\varAP$ for computation offloading,
and by $\varnumberofusers(\vardecisionsvector)\!\!=\!\!\sum_{\varAP\in\varAPs}\varnumberofusers_\varAP(\vardecisionsvector)$ the number of MUs that offload.
Similarily, for an AP $\varAP\in\varAPs$ we denote by $\varoffloaders_\varAP(\vardecisionsvector)=\{\varplayer|\vardecision_\varplayer=\varAP\}$ the set of MUs that
  offload using AP $\varAP$, and we define the set of offloaders as $\varoffloaders(\vardecisionsvector)=\cup_{\varAP\in\varAPs}\varoffloaders_\varAP(\vardecisionsvector)$.
We consider that the bandwidth $\varchannelbandwidth_\varAP$ of AP $\varAP$ is divided equally among the users that are connecting to it, i.e., the uplink rate
 $\varuplinkrate_{\varplayer}^{\varAP}(\vardecisionsvector)$ of MU $\varplayer$ is given by
 \begin{equation}
  \varuplinkrate_{\varplayer}^{\varAP}(\vardecisionsvector)=\frac{\varchannelbandwidth_\varAP}{\varnumberofusers_\varAP(\vardecisionsvector)}.
 \end{equation}
 The model of equal bandwidth sharing is reasonable if MUs are colocated, or if the APs implement
 fair uplink bandwidth allocation~\cite{Wong2009TWC,Cicalo2015eucnc}. 

 The uplink rate $\varuplinkrate_{\varplayer}^{\varAP}(\vardecisionsvector)$ together with the input data size $\vardatasize_\varplayer$  determines the transmission time $\varcloudtransmissiontime_{\varplayer,\varAP}(\vardecisionsvector)$
 of MU $\varplayer$ for offloading 
 via AP $\varAP$,   
 \begin{equation}
  \varcloudtransmissiontime_{\varplayer,\varAP}(\vardecisionsvector)=\frac{\vardatasize_{\varplayer}}{\varuplinkrate_{\varplayer}^{\varAP}(\vardecisionsvector)}.
 \end{equation}
 To model the energy consumption of the MUs, we assume that MU $\varplayer$ uses a constant transmit power of $\varpower_{\varplayer}$ for sending the data, thus the energy consumption of MU $\varplayer$ for
 offloading the input data of size $\vardatasize_{\varplayer}$ via AP $\varAP$ is\\
 \begin{equation}
  \varcloudtransmissionenergy_{\varplayer,\varAP}(\vardecisionsvector)=\frac{\vardatasize_{\varplayer}\varpower_{\varplayer}}{\varuplinkrate_{\varplayer}^{\varAP}(\vardecisionsvector)}.
 \end{equation}
\begin{figure}[tb]
\vspace{-0.2cm}
 \begin{center}
  \includegraphics[width=\columnwidth]{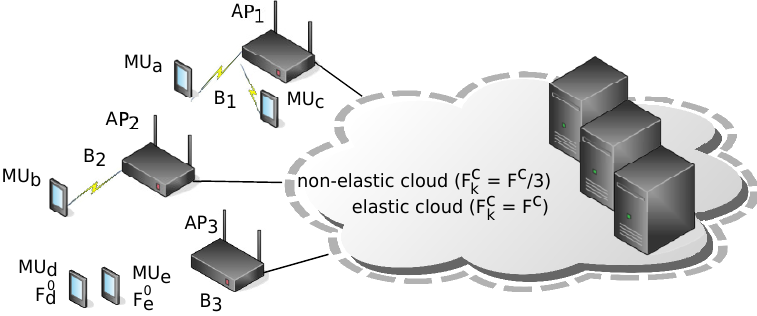}
  \caption{An example of a mobile cloud computing system}
  \label{fig:model}
  \end{center}
  \vspace{-1cm}
\hspace{0.015\textwidth}
\end{figure}
\subsection{Computation model}
 In what follows we introduce our model of the time and energy consumption of performing the computation locally and in the cloud server.

\subsubsection{Local computing}
In the case of local computing data need not be transmitted, but the task has to be processed using local
computing power. We denote by  $\varusercapability_\varplayer$ the computational capability of MU $\varplayer$,
and express the time it takes for MU $\varplayer$ to perform the computation task
 $<\!\!\vardatasize_\varplayer,\varCPUcyclesnumber_\varplayer\!\!>$ locally by
 \begin{equation}
  \varlocalcomputingtime_{\varplayer} = \frac{\varCPUcyclesnumber_{\varplayer}}{\varusercapability_{\varplayer}}.
 \end{equation}
In order to model the energy consumption of local computing we 
 denote by $\varenergyconstant_\varplayer$  the consumed energy per CPU cycle, thus we obtain
 \begin{equation}
  \varlocalcomputingenergy_{\varplayer} = \varenergyconstant_{\varplayer}\varCPUcyclesnumber_{\varplayer}.
 \end{equation}
 \subsubsection{Cloud computing}
 In the case of cloud computing, after the data are transmitted via an AP, processing is done at the cloud server.
We denote the computation capability of the cloud by $\varcloudcapability$, and by $\varcloudcapability_\varplayer$ the computation capability assigned to MU  $\varplayer$ by the cloud.
We consider two models of scaling for the computational capability of the cloud. In the \emph{elastic} model each MU that offloads receives $\varcloudcapability_\varplayer=\varcloudcapability$
amount of computing power, which is a resonable assumption for large cloud computing infrastructures.
In the \emph{non-elastic} model an MU that offloads is assigned $\varcloudcapability_\varplayer(\vardecisionsvector)=\varcloudcapability/{\varnumberofusers(\vardecisionsvector)}$ computation capability, i.e., the computing power is shared equally among all MUs that offload, which may be a reasonable model of emerging
mobile edge cloud infrastructures with limited computational power and scaling~\cite{ETSI2015mec}.

Given  $\varcloudcapability_\varplayer$ we  use a linear model to compute the execution time of a task $<\!\!\vardatasize_\varplayer,\varCPUcyclesnumber_\varplayer\!\!>$ that is offladed by MU $\varplayer$,\\
 \begin{equation}
  \varcloudcomputingtime_{\varplayer} = \frac{\varCPUcyclesnumber_{\varplayer}}{\varcloudcapability_\varplayer}.
 \end{equation}
Figure ~\ref{fig:model} shows an example of 
a mobile cloud computing system that consists of $\varAPsdim=3$ APs and $\varplayerssetdim=5$ MUs in which MUs $a$ and $c$ offload using AP 1, 
MU $b$ offloads using AP 2, and MUs $d$ and $e$ perform the local computation.  
 \subsection{Cost Model}
 We consider that the cost of an MU can be modeled as a linear combination of the time it takes to finish the computation and its energy consumption. 
 For MU $\varplayer$ we denote by $\varenergyweight_{\varplayer}$ the weight attributed to energy consumption  and by $\vartimeweight_{\varplayer}$ the weight
 attributed to the time it takes to finish the computation, $0 \leq \varenergyweight_\varplayer< \vartimeweight_\varplayer \leq 1$.

 Using these notation, for the case of local computing the cost of MU $\varplayer$ is determined by the local computing time and the corresponing energy consumption,
\begin{equation}\label{eq:local_cost_2}
  \varlocalcost_{\varplayer} = \vartimeweight_{\varplayer}\varlocalcomputingtime_{\varplayer} + \varenergyweight_{\varplayer}\varlocalcomputingenergy_{\varplayer} = (\frac{\vartimeweight_\varplayer}{\varusercapability_\varplayer}+\varenergyweight_\varplayer\varenergyconstant_\varplayer)\varCPUcyclesnumber_\varplayer.
\end{equation}

 For the case of offloading the cost is determined by the transmission time, the corresponding transmit energy, and the computing time in the cloud,
\begin{eqnarray}
  \varcloudcost_{\varplayer,\varAP}(\vardecisionsvector) &=& \vartimeweight_{\varplayer}(\varcloudcomputingtime_{\varplayer}+\varcloudtransmissiontime_{\varplayer,\varAP}(\vardecisionsvector)) + \varenergyweight_{\varplayer}\varcloudtransmissionenergy_{\varplayer,\varAP}(\vardecisionsvector)\nonumber\\
  &=&(\vartimeweight_\varplayer+\varenergyweight_\varplayer\varpower_\varplayer)\frac{\vardatasize_\varplayer}{\varuplinkrate_{\varplayer}^{\varAP}(\vardecisionsvector)} + \vartimeweight_\varplayer\frac{\varCPUcyclesnumber_\varplayer}{\varcloudcapability_\varplayer}.\label{eq:offloading_cost_2}
\end{eqnarray}
Similar to previous works~\cite{Chen2015tpds,Huang2012twc,Kumar2010Computer}, we do not model the time needed to transmit the results of the computation from the cloud server to the MU, as for typical applications like face and speech recognition, the size of the result of the computation is much smaller than $\vardatasize_\varplayer$.

For notational convenience let us define the indicator function $\varindicatorfunction(\vardecision_{\varplayer},\varAP)$ for MU $\varplayer$ as 
 \begin{equation}\label{eq:indicator_function}
  \varindicatorfunction(\vardecision_{\varplayer},\varAP)\!=\!
  \left\{\!\!\! \begin{array}{ll}
  1,& \mbox{ if }
  \vardecision_{\varplayer}=\varAP
  \\
  0,&
  \mbox{ otherwise.}\end{array} \right.
 \end{equation}
We can then express the cost of MU $\varplayer$ in strategy profile $\vardecisionsvector$ as
 \begin{equation}\label{eq:cost_function}
  \varcostfunction_{\varplayer}(\vardecisionsvector)=
    \varlocalcost_{\varplayer} \varindicatorfunction(\vardecision_{\varplayer},0)
   +
  \sum_{\varAP\in\varAPs} {\varcloudcost_{\varplayer,\varAP}(\vardecisionsvector)\varindicatorfunction(\vardecision_{\varplayer},\varAP)}.  
 \end{equation}
\subsection{Computation Offloading Game}
We consider that the objective of each MU is to minimize its cost (\ref{eq:cost_function}), i.e., to find a strategy 
\begin{equation}
  \vardecision^*_\varplayer \in \argmin_{\vardecision_\varplayer\in\vardecisionsset_\varplayer}\varcostfunction_{\varplayer}(\vardecision_\varplayer,\vardecision_{-\varplayer} ),
  \label{eq::best-reply}
\end{equation}
where we use $\vardecision_{-\varplayer}$ to denote the strategies of all MUs except MU $\varplayer$. Clearly, the strategy of an MU influences the cost of the other MUs, and thus we can model the problem as a strategic game $\varthegame=<\varplayersset, (\vardecisionsset_\varplayer)_\varplayer, (\varcostfunction_{\varplayer})_\varplayer>$, in which the players are the MUs. We refer to the game as the \emph{computation offloading game}.
We are interested in whether cost minimizing MUs can reach a strategy profile in which no MU can further decrease her cost through changing her strategy, i.e., a Nash equilibrium of the game $\varthegame$.
\begin{definition}
  A Nash equilibrium (NE) of the strategic game $<\!\!\varplayersset, (\vardecisionsset_\varplayer)_\varplayer, (\varcostfunction_{\varplayer})_\varplayer\!\!>$ is a strategy profile $\vardecision^*$ such that
  $$
\varcostfunction_{\varplayer}(\vardecision^*_\varplayer,\vardecision^*_{-\varplayer} )\leq \varcostfunction_{\varplayer}(\vardecision_\varplayer,\vardecision^*_{-\varplayer} ).
  $$
\end{definition}
Given a strategy profile $(\vardecision_\varplayer,\vardecision_{-\varplayer})$ we say that strategy $\vardecision^\prime_\varplayer$ is an improvement step for MU $\varplayer$ if $\varcostfunction_{\varplayer}(\vardecision^\prime_\varplayer,\vardecision_{-\varplayer}) < \varcostfunction_{\varplayer}(\vardecision_\varplayer,\vardecision_{-\varplayer} )$. We call a sequence of improvement steps in which one MU changes her strategy at a time an \emph{improvement path}.
Furthermore, we say that a strategy $\vardecision^*_\varplayer$ is a best reply to $ \vardecision_{-\varplayer}$ if it solves (\ref{eq::best-reply}), and we call an improvement path in which all improvement steps are best reply a \emph{best improvement path}. Observe that in a NE all MUs play their best replies to each others' strategies. 

In the rest of the paper we investigate whether NE exist for the \emph{elastic} and for the \emph{non-elastic} cloud model, and whether the MUs can compute a NE  efficiently using distributed algorithms.

%% file: equal_elastic.tex
\section{Equilibria in case of an Elastic Cloud}
\label{sec::elastic} 
Recall that under the \emph{elastic} cloud model the cloud computation 
capability assigned to user $\varplayer$ is independent of the other players' strategies, $\varcloudcapability_\varplayer=\varcloudcapability$. 
Thus, the cost function in the case of offloading can be expressed as
\begin{equation}\label{eq:offloading_cost_2ee}
  \varcloudcost_{\varplayer,\varAP}(\vardecisionsvector)=(\vartimeweight_\varplayer+\varenergyweight_\varplayer\varpower_\varplayer)\vardatasize_\varplayer\frac{\varnumberofusers_\varAP(\vardecisionsvector)}{\varchannelbandwidth_\varAP} + 
  \vartimeweight_\varplayer\frac{\varCPUcyclesnumber_\varplayer}{\varcloudcapability}.
\end{equation}

We start with formulating an insightful structural result about the best responses of the MUs, which we will use later to prove the existence of NE.
\begin{lemma}
 Given the strategy profile $\vardecision_{-\varplayer}$ of the MUs other than $\varplayer$ in the computation offloading game with elastic cloud, 
 a best reply $\vardecision_{\varplayer}^{*}$ of user $\varplayer$ satisfies the following threshold strategy\\
 \begin{equation}\label{eq:best_reply}
  \vardecision_{\varplayer}^{*}=\
  \left\{\!\!\! \begin{array}{ll}
  0,& \!\!\mbox{ if }
  \vartreshold_{\varplayer} \leq  \frac{\varnumberofusers_\varAP(\varAP,\vardecision_{-\varplayer})}{\varchannelbandwidth_\varAP} \hspace{0.1cm} \mbox{for} \hspace{0.1cm} \forall \varAP \in \varAPs
  \\
  \varAP,&
  \!\!\mbox{ if } \frac{\varnumberofusers_\varAP(\varAP,\vardecision_{-\varplayer})}{\varchannelbandwidth_\varAP} \leq \min\bigg\{\vartreshold_{\varplayer},\min\limits_{ \varoAP \in \varAPs \setminus \{\varAP\}} 
  \frac{\varnumberofusers_\varoAP(\varoAP,\vardecision_{-\varplayer})}{\varchannelbandwidth_\varoAP}\bigg\}\\
  \end{array} \right.
 \end{equation}
 where
 \begin{equation}\nonumber
  \vartreshold_{\varplayer}=\frac{\varenergyweight_\varplayer\varenergyconstant_\varplayer+\vartimeweight_\varplayer(\frac{1}{\varusercapability_\varplayer}-\frac{1}{\varcloudcapability})}{\vartimeweight_\varplayer+
  \varenergyweight_\varplayer\varpower_\varplayer}\cdot\frac{\varCPUcyclesnumber_\varplayer}{\vardatasize_\varplayer}.
 \end{equation}
\end{lemma}
\begin{proof}
Based on~(\ref{eq:local_cost_2}),~(\ref{eq:indicator_function}),~(\ref{eq:cost_function}) and~(\ref{eq:offloading_cost_2ee}), the cost of MU $\varplayer$ when choosing $\vardecision_\varplayer$ is
\begin{align}
 & \nonumber \varcostfunction_{\varplayer}(\vardecision_\varplayer, \vardecision_{-\varplayer}) =\varlocalcost_{\varplayer}\varindicatorfunction(\vardecision_{\varplayer},0) + \displaystyle{\sum_{i=1}^{\varAPsdim}\varcloudcost_{\varplayer,\varAP}(\vardecision_\varplayer, \vardecision_{-\varplayer})\varindicatorfunction(\vardecision_{\varplayer},\varAP)}\\ \nonumber
 & = \big((\frac{\vartimeweight_\varplayer}{\varusercapability_\varplayer}+\varenergyweight_\varplayer\varenergyconstant_\varplayer)\varCPUcyclesnumber_\varplayer\big)\varindicatorfunction(\vardecision_{\varplayer},0)\\ 
 & +\sum\limits_{i=1}^{\varAPsdim}\big((\vartimeweight_\varplayer+\varenergyweight_\varplayer\varpower_\varplayer)\vardatasize_\varplayer\frac{\varnumberofusers_\varAP(\varAP,\vardecision_{-\varplayer})}{\varchannelbandwidth_\varAP} +  \vartimeweight_\varplayer\frac{\varCPUcyclesnumber_\varplayer}{\varcloudcapability}\big)\varindicatorfunction(\vardecision_{\varplayer},\varAP). \nonumber
\end{align}
Let us first consider the case that the best reply of MU $\varplayer$ is $\vardecision_{\varplayer}^{*}=0$. We then have that
$\varcostfunction_{\varplayer}(0,\vardecision_{-\varplayer}) \leq \varcostfunction_{\varplayer}(\varAP,\vardecision_{-\varplayer})$ for every AP $\varAP \in \varAPs$, which implies that
\begin{align}
\nonumber (\frac{\vartimeweight_\varplayer}{\varusercapability_\varplayer}+\varenergyweight_\varplayer\varenergyconstant_\varplayer)\varCPUcyclesnumber_\varplayer \leq 
(\vartimeweight_\varplayer+\varenergyweight_\varplayer\varpower_\varplayer)\vardatasize_\varplayer\frac{\varnumberofusers_\varAP(\varAP,\vardecision_{-\varplayer})}{\varchannelbandwidth_\varAP} + 
\vartimeweight_\varplayer\frac{\varCPUcyclesnumber_\varplayer}{\varcloudcapability}.
\end{align}
After algebraic manipulations we obtain\\
\begin{align}
\nonumber \vartreshold_\varplayer \triangleq \frac{\varenergyweight_\varplayer\varenergyconstant_\varplayer+\vartimeweight_\varplayer(\frac{1}{\varusercapability_\varplayer}-\frac{1}{\varcloudcapability})}{\vartimeweight_\varplayer+
\varenergyweight_\varplayer\varpower_\varplayer}\cdot\frac{\varCPUcyclesnumber_\varplayer}{\vardatasize_\varplayer} \leq \frac{\varnumberofusers_\varAP(\varAP,\vardecision_{-\varplayer})}{\varchannelbandwidth_\varAP}.
\end{align}
Let us now consider the case when the best reply of MU $\varplayer$  is $\vardecision_{\varplayer}^{*}=\varAP$. We then have that $\varcostfunction_{\varplayer}(\varAP,\vardecision_{-\varplayer}) \leq \varcostfunction_{\varplayer}(0,\vardecision_{-\varplayer})$ and 
$\varcostfunction_{\varplayer}(\varAP,\vardecision_{-\varplayer}) \leq \varcostfunction_{\varplayer}(\varoAP,\vardecision_{-\varplayer})$ for every AP $\varoAP \in \varAPs\setminus\{\varAP\}$. Following the same reasoning as above, it is easy to see that 
$\varcostfunction_{\varplayer}(\varAP,\vardecision_{-\varplayer}) \leq \varcostfunction_{\varplayer}(0,\vardecision_{-\varplayer}^{*})$ implies that $\frac{\varnumberofusers_\varAP(\varAP,\vardecision_{-\varplayer})}{\varchannelbandwidth_\varAP} \leq \vartreshold_{\varplayer}$. 
It is easy to see that $\varnumberofusers_\varoAP(\varoAP,\vardecision_{-\varplayer})=\varnumberofusers_\varoAP(\varAP,\vardecision_{-\varplayer})+1$,
and thus $\varcostfunction_{\varplayer}(\varAP,\vardecision_{-\varplayer}^{*}) \leq \varcostfunction_{\varplayer}(\varoAP,\vardecision_{-\varplayer}^{*})$ implies that
\begin{align}
  \nonumber (\vartimeweight_\varplayer+\varenergyweight_\varplayer\varpower_\varplayer)\vardatasize_\varplayer\frac{\varnumberofusers_\varAP(\varAP,\vardecision_{-\varplayer})}{\varchannelbandwidth_\varAP}
  \leq 
  (\vartimeweight_\varplayer+\varenergyweight_\varplayer\varpower_\varplayer)\vardatasize_\varplayer\frac{\varnumberofusers_{\varoAP}(\varoAP,\vardecision_{-\varplayer})}{\varchannelbandwidth_\varoAP}
\end{align}
which is equivalent to
\begin{align}
\nonumber \frac{\varnumberofusers_\varAP(\varAP,\vardecision_{-\varplayer})}{\varchannelbandwidth_\varAP} \leq \frac{\varnumberofusers_\varoAP(\varoAP,\vardecision_{-\varplayer})}{\varchannelbandwidth_\varoAP}.
\end{align}
\end{proof}
The above threshold strategy allows players to compute their best and better replies efficiently. In what follows we show
that the computation offloading game with elastic cloud admits a NE, and a NE can be computed by iterative computation of the players' better or best replies, i.e., following an improvement path.
Before we formulate the theorem, let us recall the definition of a generalized ordinal potential from~\cite{Monderer1996124}.

\begin{definition}
  A function $\Phi:\times\vardecisionsset_\varplayer\rightarrow \mathbb{R}$ is a generalized ordinal potential function for the strategic game $<\varplayersset, (\vardecisionsset_\varplayer)_\varplayer, (\varcostfunction_{\varplayer})_\varplayer>$ if for an arbitrary strategy profile $(\vardecision_\varplayer,\vardecision_{-\varplayer})$ and for any corresponding improvement step $\vardecision^\prime_\varplayer$ it holds that
  \begin{align} 
  \varcostfunction_{\varplayer}(\vardecision^\prime_\varplayer,\vardecision_{-\varplayer}) -  &\varcostfunction_{\varplayer}(\vardecision_\varplayer,\vardecision_{-\varplayer}) < 0 \Rightarrow \nonumber\\
    & \Phi(\vardecision^\prime_\varplayer,\vardecision_{-\varplayer}) -  \Phi(\vardecision_\varplayer,\vardecision_{-\varplayer}) < 0.\nonumber\end{align}
  \end{definition}
 
\begin{theorem}
  The computation offloading game with elastic cloud admits the generalized ordinal potential function
\begin{equation}\label{eq:potential_function}
    \Phi(\vardecisionsvector)= \displaystyle{\sum_{m=1}^{\varAPsdim} \displaystyle{\sum_{n=1}^{\varnumberofusers_m(\vardecisionsvector)}}}
    \frac{n}{\varchannelbandwidth_m}
    + \displaystyle{\sum_{s=1}^{\varplayerssetdim}}\vartreshold_{s}\varindicatorfunction(\vardecision_s,0),
   \end{equation}
  and hence it possesses a pure strategy Nash equilibrium. 
\end{theorem}
\begin{proof}
  To prove that $\Phi(\vardecisionsvector)$ is a generalized ordinal potential function, we first show
  that $\varcostfunction_{\varplayer}(\varAP,\vardecision_{-\varplayer})<\varcostfunction_{\varplayer}(0,\vardecision_{-\varplayer})$
 implies $\Phi_{\varplayer}(\varAP,\vardecision_{-\varplayer})<\Phi_{\varplayer}(0,\vardecision_{-\varplayer})$ for a MU $\varplayer$.
 According to~(\ref{eq:local_cost_2}),~(\ref{eq:cost_function}) and~(\ref{eq:offloading_cost_2ee}), the condition 
 $\varcostfunction_{\varplayer}(\varAP,\vardecision_{-\varplayer})<\varcostfunction_{\varplayer}(0,\vardecision_{-\varplayer})$ implies that 
 \begin{equation}\label{eq:condition_1}
 \frac{\varnumberofusers_\varAP(\varAP,\vardecision_{-\varplayer})}{\varchannelbandwidth_\varAP} < \vartreshold_{\varplayer}
 \end{equation}
for the strategy profile $(\varAP,\vardecision_{-\varplayer})$  it holds that
 \begin{align}
\nonumber \Phi(\varAP,\vardecision_{-\varplayer}) = \!\!\sum\limits_{n=1}^{\varnumberofusers_\varAP(\varAP,\vardecision_{-\varplayer})}\!\!\frac{n}{\varchannelbandwidth_\varAP}
+\!\!\sum\limits_{m\neq\varAP} \!\!\!\!\sum\limits_{n=1}^{\varnumberofusers_m(\varAP,\vardecision_{-\varplayer})}\!\!\frac{n}{\varchannelbandwidth_m}
+ \!\!\sum\limits_{s\neq\varplayer}\vartreshold_{s}\varindicatorfunction(\vardecision_s,0),
 \end{align}
and for the strategy profile $(0,\vardecision_{-\varplayer})$
 \begin{align}
\nonumber \Phi(0,\vardecision_{-\varplayer}) = \!\!\!\!\!\!\sum\limits_{n=1}^{\varnumberofusers_\varAP(0,\vardecision_{-\varplayer})}\!\!\!\!\!\!\!\frac{n}{\varchannelbandwidth_\varAP}
+\!\!\sum\limits_{m\neq\varAP}\!\!\!\!\! \sum\limits_{n=1}^{\varnumberofusers_m(0,\vardecision_{-\varplayer})}\!\!\!\!\!\!\!\!\frac{n}{\varchannelbandwidth_m}
 + \vartreshold_\varplayer + \!\!\sum\limits_{s\neq\varplayer}\!\vartreshold_{s}\varindicatorfunction(\vardecision_s,0).
 \end{align}
Since $\varnumberofusers_\varAP(\varAP,\vardecision_{-\varplayer})=\varnumberofusers_\varAP(0,\vardecision_{-\varplayer})+1$,  we obtain
\begin{align}
 \nonumber \Phi(\varAP,\vardecision_{-\varplayer}) - \Phi(0,\vardecision_{-\varplayer}) = \frac{\varnumberofusers_\varAP(\varAP,\vardecision_{-\varplayer})}{\varchannelbandwidth_\varAP}-\vartreshold_\varplayer.
\end{align}
It follows from~(\ref{eq:condition_1}) that $\Phi(\varAP,\vardecision_{-\varplayer}) - \Phi(0,\vardecision_{-\varplayer})<0$. 
Similarly, we can show that $\varcostfunction_{\varplayer}(0,\vardecision_{-\varplayer})<\varcostfunction_{\varplayer}(\varAP,\vardecision_{-\varplayer})$
implies $\Phi_{\varplayer}(0,\vardecision_{-\varplayer})<\Phi_{\varplayer}(\varAP,\vardecision_{-\varplayer})$.

Second, we show that $\varcostfunction_{\varplayer}(\varAP,\vardecision_{\varplayer})<\varcostfunction_{\varplayer}(\varoAP,\vardecision_{\varplayer})$
implies $\Phi_{\varplayer}(\varAP,\vardecision_{\varplayer})<\Phi_{\varplayer}(\varoAP,\vardecision_{\varplayer})$ for a MU $\varplayer$.
According to~(\ref{eq:cost_function}) and~(\ref{eq:offloading_cost_2ee}), the condition 
$\varcostfunction_{\varplayer}(\varAP,\vardecision_{\varplayer})<\varcostfunction_{\varplayer}(\varoAP,\vardecision_{\varplayer})$ implies that 
\begin{equation}\label{eq:condition_2}
 \frac{\varnumberofusers_\varAP(\varAP,\vardecision_{\varplayer})}{\varchannelbandwidth_\varAP} < \frac{\varnumberofusers_\varoAP(\varoAP,\vardecision_{\varplayer})}{\varchannelbandwidth_\varoAP}
\end{equation}
Let us rewrite $\Phi$ by separating the terms for APs $\varAP$ and $\varoAP$,
\begin{align}
\nonumber \Phi(\varAP,\vardecision_{-\varplayer}) & = \displaystyle{\sum_{n=1}^{\varnumberofusers_\varAP(\varAP,\vardecision_{-\varplayer})}}\frac{n}{\varchannelbandwidth_\varAP} +
\displaystyle{\sum_{n=1}^{\varnumberofusers_\varoAP(\varAP,\vardecision_{-\varplayer})}}\frac{n}{\varchannelbandwidth_\varoAP}\\
& +\displaystyle{\sum_{m\neq\varAP,\varoAP} \displaystyle{\sum_{n=1}^{\varnumberofusers_m}}}\frac{n}{\varchannelbandwidth_m}
+ \displaystyle{\sum_{s\neq\varplayer}}\vartreshold_{s}\varindicatorfunction(\vardecision_s,0).
 \end{align}
Since $\varnumberofusers_\varAP(\varAP,\vardecision_{-\varplayer})=\varnumberofusers_\varAP(\varoAP,\vardecision_{-\varplayer})+1$ and $\varnumberofusers_\varoAP(\varoAP,\vardecision_{-\varplayer})=\varnumberofusers_\varoAP(\varAP,\vardecision_{-\varplayer})+1$, we have that 
\begin{align}
 \nonumber \Phi(\varAP,\vardecision_{-\varplayer}) - \Phi(\varoAP,\vardecision_{-\varplayer})=\frac{\varnumberofusers_\varAP(\varAP,\vardecision_{-\varplayer})}{\varchannelbandwidth_\varAP}-
 \frac{\varnumberofusers_\varoAP(\varoAP,\vardecision_{-\varplayer})}{\varchannelbandwidth_\varoAP}.
\end{align}
It follows from~(\ref{eq:condition_2}) that $\Phi(\varAP,\vardecision_{-\varplayer}) - \Phi(\varoAP,\vardecision_{-\varplayer})<0$, which proves the theorem.
\end{proof}
The existence of a generalized ordinal potential allows us to formulate a simple algorithm for computing a Nash equilibrium by leveraging the fact that
in a game that all improvement paths are finite, i.e., lead to a Nash equilibrium, in a finite strategic game that admits a generalized ordinal potential function~\cite{Monderer1996124}. 
\begin{corollary}
  \label{th::elastic-fip}
 Starting from an arbitrary initial strategy profile, let one MU at a time perform an improvement step iteratively. The algorithm terminates in a NE after a finite number of steps for the computation offloading game with elastic cloud. 
\end{corollary}

%% file: equal_non_elastic.tex
\section{Equilibria in case of a Non-Elastic Cloud}
 \label{sec::non-elastic}
In the case of a non-elastic cloud the  computation 
capability $\varcloudcapability_\varplayer$ that is assigned to MU $\varplayer$ in the cloud server depends on the other MUs' strategies, and thus the
 cost function in case of offloading can be expressed as
\begin{equation}\label{eq:offloading_cost_2ene}
  \varcloudcost_{\varplayer,\varAP}(\vardecisionsvector)=(\vartimeweight_\varplayer+\varenergyweight_\varplayer\varpower_\varplayer)\vardatasize_\varplayer\frac{\varnumberofusers_\varAP(\vardecisionsvector)}{\varchannelbandwidth_\varAP} + 
  \vartimeweight_\varplayer\frac{\varCPUcyclesnumber_\varplayer}{\varcloudcapability}\varnumberofusers(\vardecisionsvector).
\end{equation}
A natural question is whether a generalized ordinal potential similar to (\ref{eq:potential_function}) exists in the case of non-elastic cloud,
in which case all improvement paths would be finite. We first show that if we only allow MUs to change between APs, but we do not allow them to start or to stop offloading, then this holds.
\begin{lemma}
  \label{th:swapping-finite}
  Consider an arbitrary strategy profile $\vardecisionsvector$, and consider that (i) the improvement step $\vardecision^\prime_\varplayer$ of MU $\varplayer\in \varoffloaders(\vardecisionsvector)$ is constrained to $\vardecision^\prime_\varplayer\in \varAPs$,
  and (ii) MUs $\varplayer\not\in \varoffloaders(\vardecisionsvector)$ are not allowed to perform improvement steps.
  Then all improvement paths that satisfy constraints (i) and (ii) are finite.
\end{lemma}
\begin{proof}
  First, observe that the set $\varoffloaders(\vardecisionsvector)$  of offloaders is unchanged during an improvement path with constraints (i) and (ii). For a strategy profile $\vardecisionsvector$ let
  vector $\varcostvector(\vardecisionsvector)\in\mathbb{R}_{\geq 0}^{\vert \varoffloaders(\vardecisionsvector) \vert}$ contain the cost $\varcostfunction_{\varplayer}(\vardecisionsvector)$ for MUs $\varplayer \in \varoffloaders(\vardecisionsvector)$ in decreasing order. Let $\vardecisionsvector^\prime=(\vardecision_\varplayer^\prime,\vardecision_{-\varplayer})$ be the strategy profile after an improvement step made by MU $\varplayer$ that satisfies constraint (i), and let $\varAP=\vardecision_\varplayer$ and $\varoAP=\vardecision_\varplayer^\prime$.
  Since  $\frac{\varnumberofusers_\varoAP(\vardecisionsvector)+1}{\varchannelbandwidth_\varoAP}<\frac{\varnumberofusers_\varAP(\vardecisionsvector)}{\varchannelbandwidth_\varAP}$ must hold for the change of APs to be an improvement step, we have $\varcostvector(\vardecisionsvector^\prime)\prec_{L}\varcostvector(\vardecisionsvector)$, where $\prec_{L}$ stands for lexicographically smaller. Since $\varcostvector(\vardecisionsvector)$ decreases in the lexiographical sense upon every improvement step, and the number of strategy profiles is finite, the improvement paths must be finite.
\end{proof}
Thus, if MUs can only change between APs, they terminate after a finite number of improvement steps.
Unfortunately, as the following example shows, this is not the case if MUs can decide not to offload, and thus the
algorithm in Corrollary~\ref{th::elastic-fip} cannot be used to compute a NE, even if a NE exists.
\begin{example}
  Consider a computation offloading game with non-elastic cloud where $\varplayersset=\{a,b,c,d,e\}$ and $\varAPs=\{1,2,3\}$ as illustrated in Figure~\ref{fig:model}.
  Figure~\ref{fig::example_cycle} shows a cyclic improvement path starting from the strategy profile  $(1,2,1,0,0)$,  in which MUs $a$ and $c$ are connected
to AP 1, MU $b$ is connected to AP 2 and MUs $d$ and $e$ perform local computation.
\end{example}
\begin{figure}[h!]
\pgfdeclarelayer{background}
\pgfdeclarelayer{foreground}
\pgfsetlayers{background,main,foreground}
\centering
\vspace{-0.5cm}
\begin{tikzpicture}
\tikzset{ 
table/.style={
  matrix of nodes,
  row sep=-\pgflinewidth,
  column sep=-\pgflinewidth,
  nodes={rectangle,text width=0.7em,align=center},
  text depth=1.25ex,
  text height=2ex,
  nodes in empty cells
},
row 1/.style={nodes={fill=green!10}},
column 1/.style={nodes={fill=green!10,text width=1.6em}},
}
\matrix (mat) [table]
{
    $\vardecision_\varplayer$ & $\vardecision_a$ & $\vardecision_b$ & $\vardecision_c$ & $\vardecision_d$ & $\vardecision_e$\\ 
    $\vardecisionsvector(0)$ & 1 & 2 & 1 & 0 & 0 \\
    $\vardecisionsvector(1)$ & 1 & 2 & 2 & 0 & 0 \\
    $\vardecisionsvector(2)$ & 1 & 0 & 2 & 0 & 0 \\
    $\vardecisionsvector(3)$ & 1 & 0 & 2 & 2 & 0 \\
    $\vardecisionsvector(4)$ & 1 & 0 & 2 & 2 & 2 \\
    $\vardecisionsvector(5)$ & 1 & 0 & 1 & 2 & 2 \\
    $\vardecisionsvector(6)$ & 1 & 3 & 1 & 2 & 2 \\
    $\vardecisionsvector(7)$ & 1 & 3 & 1 & 2 & 0 \\
    $\vardecisionsvector(8)$ & 1 & 3 & 1 & 0 & 0 \\
    $\vardecisionsvector(9)$ & 1 & 2 & 1 & 0 & 0 \\
};
\draw 
    ([xshift=-.5\pgflinewidth]mat-1-1.north west) --   
    ([xshift=-.5\pgflinewidth]mat-1-6.north east);
\draw [thick,double]
    ([xshift=-.5\pgflinewidth]mat-1-1.south west) --   
    ([xshift=-.5\pgflinewidth]mat-1-6.south east);
\draw 
    ([xshift=-.5\pgflinewidth]mat-11-1.south west) --   
    ([xshift=-.5\pgflinewidth]mat-11-6.south east);
\draw 
    ([yshift=.10\pgflinewidth]mat-1-1.north west) -- 
    ([yshift=.10\pgflinewidth]mat-11-1.south west);
\draw 
    ([yshift=.10\pgflinewidth]mat-1-6.north east) -- 
    ([yshift=.10\pgflinewidth]mat-11-6.south east);
    
\foreach \x in {2,...,10}
{
  \draw 
    ([xshift=-.5\pgflinewidth]mat-\x-1.south west) --   
    ([xshift=-.5\pgflinewidth]mat-\x-6.south east);
  }
  
\foreach \x in {1,...,5}
{
  \draw 
    ([yshift=.10\pgflinewidth]mat-1-\x.north east) -- 
    ([yshift=.10\pgflinewidth]mat-11-\x.south east);
}


\begin{scope}[shorten >=7pt,shorten <= 7pt]
\draw[->]  (mat-2-4.center) -- (mat-3-4.center);
\draw[->]  (mat-3-3.center) -- (mat-4-3.center);
\draw[->]  (mat-4-5.center) -- (mat-5-5.center);
\draw[->]  (mat-5-6.center) -- (mat-6-6.center);
\draw[->]  (mat-6-4.center) -- (mat-7-4.center);
\draw[->]  (mat-7-3.center) -- (mat-8-3.center);
\draw[->]  (mat-8-6.center) -- (mat-9-6.center);
\draw[->]  (mat-9-5.center) -- (mat-10-5.center);
\draw[->]  (mat-10-3.center) -- (mat-11-3.center);
\end{scope}

\node at ([xshift=3.1cm, yshift=-0.48cm]mat-2-6.east) 
{$\scriptstyle {\varchannelbandwidth_2>\varchannelbandwidth_1 \hspace{2.95cm}\footnotesize  (1)}$};
\node at ([xshift=3cm, yshift=-0.48cm]mat-3-6.east) 
{$\scriptstyle \frac{2}{\varchannelbandwidth_2}(\vartimeweight_b+\varenergyweight_b\varpower_b)\vardatasize_b+
3\vartimeweight_b\frac{\varCPUcyclesnumber_b}{\varcloudcapability}>\varlocalcost_b\hspace{0.09cm}\footnotesize  (2)$};
\node at ([xshift=3cm, yshift=-0.48cm]mat-4-6.east) 
{$\scriptstyle \varlocalcost_d>\frac{2}{\varchannelbandwidth_2}(\vartimeweight_d+\varenergyweight_d\varpower_d)\vardatasize_d+
3\vartimeweight_d\frac{\varCPUcyclesnumber_d}{\varcloudcapability}\hspace{0.08cm}\footnotesize  (3)$};
\node at ([xshift=3cm, yshift=-0.48cm]mat-5-6.east) 
{$\scriptstyle \varlocalcost_e>\frac{3}{\varchannelbandwidth_2}(\vartimeweight_e+\varenergyweight_e\varpower_e)\vardatasize_e+
4\vartimeweight_e\frac{\varCPUcyclesnumber_e}{\varcloudcapability}\hspace{0.08cm}\footnotesize  (4)$};
\node at ([xshift=3.1cm, yshift=-0.48cm]mat-6-6.east) 
{$\scriptstyle \varchannelbandwidth_1>\frac{2}{3}\varchannelbandwidth_2\hspace{2.7cm}\footnotesize  (5)$};
\node at ([xshift=3cm, yshift=-0.48cm]mat-7-6.east) 
{$\scriptstyle \varlocalcost_b>\frac{1}{\varchannelbandwidth_3}(\vartimeweight_b+\varenergyweight_b\varpower_b)\vardatasize_b+
5\vartimeweight_b\frac{\varCPUcyclesnumber_b}{\varcloudcapability}\hspace{0.08cm}\footnotesize  (6)$};
\node at ([xshift=3cm, yshift=-0.48cm]mat-8-6.east) 
{$\scriptstyle \frac{2}{\varchannelbandwidth_2}(\vartimeweight_e+\varenergyweight_e\varpower_e)\vardatasize_e+
5\vartimeweight_e\frac{\varCPUcyclesnumber_e}{\varcloudcapability}>\varlocalcost_e\hspace{0.08cm}\footnotesize  (7)$};
\node at ([xshift=3cm, yshift=-0.48cm]mat-9-6.east) 
{$\scriptstyle \frac{1}{\varchannelbandwidth_2}(\vartimeweight_d+\varenergyweight_d\varpower_d)\vardatasize_d+
4\vartimeweight_d\frac{\varCPUcyclesnumber_d}{\varcloudcapability}>\varlocalcost_d\hspace{0.08cm}\footnotesize  (8)$};
\node at ([xshift=3.1cm, yshift=-0.48cm]mat-10-6.east) 
{$\scriptstyle \varchannelbandwidth_2>\varchannelbandwidth_3\hspace{2.9cm}\footnotesize  (9)$};
\end{tikzpicture}
\vspace{-0.5cm}
\caption{A cyclic improvement path in a computation offloading game with non-elastic cloud, 3 APs and 5 MUs. Rows correspond to strategy profiles, columns to MUs. An arrow between adjacent rows indicates the MU that performs the improvement step. The cycle consists of 9 improvement steps, and involves some MUs to start and to stop offloading. The inequalities on the right show the condition under which the change of strategy is an improvement step.}
\label{fig::example_cycle}
\end{figure}
Starting from the initial strategy profile $(1,2,1,0,0)$, Player $c$ revises its strategy to AP $2$, which is an improvement step
if $\varchannelbandwidth_2>\varchannelbandwidth_1$, as shown in inequality (1) in the figure.
Observe that after $9$ improvement steps the players reach the initial strategy profile.
For each step the inequality on the right provides the condition for being an improvement. It follows from inequalities (1), (5) and (9) that 
$\varchannelbandwidth_2>\varchannelbandwidth_1$, $\varchannelbandwidth_1>\frac{2}{3}\varchannelbandwidth_2$ and
$\varchannelbandwidth_2>\varchannelbandwidth_3$, respectively. 
Since,$\frac{1}{\varchannelbandwidth_3}(\vartimeweight_b+\varenergyweight_b\varpower_b)\vardatasize_b+
5\vartimeweight_b\frac{\varCPUcyclesnumber_b}{\varcloudcapability}>\frac{1}{\varchannelbandwidth_3}(\vartimeweight_b+\varenergyweight_b\varpower_b)\vardatasize_b+
3\vartimeweight_b\frac{\varCPUcyclesnumber_b}{\varcloudcapability}$ holds, from inequalities (2) and (6) follows that $\varchannelbandwidth_3>\frac{1}{2}\varchannelbandwidth_2$. 
Combining inequalities (3) and (8) we have that $\vartimeweight_d\frac{\varCPUcyclesnumber_d}{\varcloudcapability}>\frac{1}{\varchannelbandwidth_2}(\vartimeweight_d+\varenergyweight_d\varpower_d)\vardatasize_d$. 
Similarly, it follows from inequalities (4) and (7) that $\vartimeweight_e\frac{\varCPUcyclesnumber_e}{\varcloudcapability}>\frac{1}{\varchannelbandwidth_2}(\vartimeweight_e+\varenergyweight_e\varpower_e)\vardatasize_e$.
Given these constraints, an instance of the example can be formulated easily.

An important consequence of the cycle in the improvement path is that the computation offloading game with non-elastic cloud does not allow a potential function, and thus Corollary~\ref{th::elastic-fip} cannot be applied. Yet, as we now show, NE always exist.
\begin{theorem}\label{theo:NE_ene}
The computation offloading game with non-elastic cloud possesses a pure strategy Nash equilibrium. 
\end{theorem}
\begin{proof}
We use induction in the number $\varplayerssetdim$ of players in order to prove the theorem, and we denote by $\varplayerssetdim^{(\vargamestage)}=\vargamestage$ the number of MUs that are involved
in the game in induction step $\vargamestage$.

It is clear that for $\varplayerssetdim^{(1)}=1$ there is a NE, in which the only participating MU plays her best reply $\vardecision_\varplayer^*(1)$. Since there
 are no other MUs, $\vardecisionsvector^*(1)$ is a NE. Observe that if $\vardecision_\varplayer^*(1)=0$, MU $\varplayer$ would never have an incentive
to deviate from this decision, because the number of players that offload will not decrease as more MUs are added. Otherwise, if MU $\varplayer$ decides to 
offload, her best reply is given by $\vardecision_\varplayer^*(1)=\argmax_{\varAP\in\varAPs}\varchannelbandwidth_\varAP$.

Assume now that for $\vargamestage-1 > 0$ there is a NE $\vardecisionsvector^*(\vargamestage-1)$.
Upon induction step $\vargamestage$ one MU enters the game; we refer to this MU as MU $\varplayerssetdim^{(\vargamestage)}$.
Let MU $\varplayerssetdim^{(\vargamestage)}$ play her best reply $\vardecision_{\varplayerssetdim^{(\vargamestage)}}^*(\vargamestage)$
with respect to the NE strategy profile of the MUs that already participated in induction step $\vargamestage-1$, i.e.,
with respect to $\vardecision_{-\varplayerssetdim^{(\vargamestage)}}(\vargamestage)=\vardecisionsvector^*(\vargamestage-1)$.
After that, MUs can perform best improvement steps one at a time starting from the strategy profile 
$\vardecisionsvector(\vargamestage)=(\vardecision_{\varplayerssetdim^{(\vargamestage)}}^*(\vargamestage),\vardecision_{-\varplayerssetdim^{(\vargamestage)}}(\vargamestage))$,
following the algorithm shown in Figure~\ref{fig:uphase}.
We refer to this as the update phase. In order to prove that there is a NE in induction step $t$, in the following we show that the MUs will perform 
a finite number of best improvement steps in the update phase. 

Let us define the reluctance to offload via AP $\varAP$ of MU $\varplayer$ in a strategy profile $\vardecisionsvector(\vargamestage)$ as 
$\varreluctanceratio_\varplayer(\vardecisionsvector(\vargamestage))=\frac{\varcloudcost_{\varplayer,\varAP}(\vardecisionsvector(\vargamestage))}{\varlocalcost_{\varplayer}}$, and let us
rank the MUs that play the same strategy in decreasing order of reluctance.
We use the triplet $(\vargamestage,l,\varAP)$ to index the MU that in step $\vargamestage$ occupies position $l$ in the ranking for AP $\varAP$, i.e.,
$ \varreluctanceratio_{(\vargamestage,1,\varAP)}(\vardecisionsvector(\vargamestage)) \geq \varreluctanceratio_{(\vargamestage,2,\varAP)}(\vardecisionsvector(\vargamestage)) 
\geq \ldots \geq \varreluctanceratio_{(\vargamestage,\varnumberofusers_\varAP(\vardecisionsvector(\vargamestage)),\varAP)}(\vardecisionsvector(\vargamestage))$.
Note that for AP $\varAP$ it is MU $(\vargamestage,1,\varAP)$ that can gain most by changing her strategy
among all  MUs $\varplayer\in\varoffloaders_\varAP(\vardecisionsvector(\vargamestage))$.

Observe that if $\vardecision^*_{\varplayerssetdim^{(\vargamestage)}}(\vargamestage)=0$, then
$\varnumberofusers_\varAP(\vardecisionsvector(\vargamestage))=\varnumberofusers_\varAP(\vardecisionsvector^*(\vargamestage-1))$ for every $\varAP \in \varAPs$ and thus $\vardecisionsvector(t)$ is a NE. 
If $\vardecision^*_{\varplayerssetdim^{(\vargamestage)}}(\vargamestage)=\varAP\in\varAPs$, but none of the MUs want to deviate from their strategy in $\vardecisionsvector^*(\vargamestage-1)$
then $\vardecisionsvector(t)$ is a NE. Otherwise, we can have one or both of the following: 
\begin{inparaenum}[(i)]
 \item for some MUs $\varplayer\in\varoffloaders_\varAP(\vardecisionsvector(\vargamestage))$ offloading using AP $\varAP$ is not a best reply anymore,
 \item for some MUs $\varplayer\in\varoffloaders_\varoAP(\vardecisionsvector(\vargamestage))$ for $\varoAP \in \varAPs \setminus \{\varAP\}$, offloading using AP $\varoAP$ is not a best reply anymore. 
\end{inparaenum} Let us denote by $\vardeviatorsOtoL$ the set of APs with at least one MU that wants to deviate from her strategy either for (i) or for (ii).

Observe that case (i) can happen only if $\varreluctanceratio_{(\vargamestage,1,\varAP)}(\vardecisionsvector(\vargamestage)) > \varreluctanceratio_{\varplayerssetdim^{(\vargamestage)}}(\vardecisionsvector(\vargamestage))$,
as otherwise no MU $\varplayer\in\varoffloaders_\varAP(\vardecisionsvector(\vargamestage))$ would  be able to gain by changing her strategy from AP $\varAP$.
Now, since $\vardecision_{\varplayerssetdim^{(\vargamestage)}}(\vargamestage)=\varAP$ it follows that
$\frac{\varnumberofusers_\varAP(\vardecisionsvector^*(\vargamestage-1))+1}{\varchannelbandwidth_\varAP}\leq
\frac{\varnumberofusers_\varoAP(\vardecisionsvector^*(\vargamestage-1))+1}{\varchannelbandwidth_\varoAP}$ for every $\varoAP \in \varAPs \setminus \{\varAP\}$.
Therefore, in case (i) an MU $\varplayer\in\varoffloaders_\varAP(\vardecisionsvector(\vargamestage))$ cannot decrease her offloading cost by choosing another AP $\varoAP$;
as an improvement step she would change her strategy to local computing.
Let now MU $(\vargamestage,1,\varAP)$ perform an improvement step, and
let us denote the resulting strategy profile by $\vardecisionsvector^\prime(\vargamestage)$ (Line $4$).
Since MU $(\vargamestage,1,\varAP)$ changed from AP $\varAP$ to local computation, $\varnumberofusers_\varAP(\vardecisionsvector^\prime(\vargamestage))=\varnumberofusers_\varAP(\vardecisionsvector^*(\vargamestage-1))$ would hold for every $\varAP \in \varAPs$ and $\vardecisionsvector^\prime(\vargamestage)$ would be a NE.

Let us now consider case (ii). The only reason why case (ii) could happen is that the number of players that offload was incremented,
i.e.,  $\varnumberofusers(\vardecisionsvector(\vargamestage))=\varnumberofusers(\vardecisionsvector^*(\vargamestage-1))+1$. Thus, the best improvement  
of every MU $\varplayer\in\varoffloaders_\varoAP(\vardecisionsvector(\vargamestage))$ that wants to deviate would be to perform the computation locally.
Among all MUs that would like to deviate, let us choose the MU with highest reluctance $\varreluctanceratio_\varplayer(\vardecisionsvector(\vargamestage))$ (note that this is MU
$(\vargamestage,1,\varoAP)$ for some $\varoAP\not=\varAP$), and let her perform the improvement step, i.e., change to local computation (Lines $9-12$).
Let the resulting strategy profile be $\vardecisionsvector^{\prime}(\vargamestage)$.
Due to this improvement step $\varnumberofusers_\varoAP(\vardecisionsvector^{\prime}(\vargamestage))=\varnumberofusers_\varoAP(\vardecisionsvector^*(\vargamestage-1))-1$,
and thus some MUs may be able to decrease their cost by connecting to AP $\varoAP$.
If there is no MU $\varplayer\in\varplayersset\setminus\varoffloaders(\vardecisionsvector^{\prime}(\vargamestage))$ that would like to start offloading,
then there is no more MU that would like to stop offloading either because $\varnumberofusers(\vardecisionsvector^{\prime}(\vargamestage))=\varnumberofusers(\vardecisionsvector^*(\vargamestage-1))-1$.
Otherwise, among all MUs $\varplayer\in\varplayersset\setminus\varoffloaders(\vardecisionsvector^{\prime}(\vargamestage))$ that would like to start offloading,
let MU $\varoplayer$ with highest local computing cost $\varlocalcost_{\varplayer^{\prime}}$ perform an improvement step, i.e., connect to AP $\varoAP$.
We now repeat these steps starting from Line $8$ until no more MU wants to stop offloading. This iteration will stop after a finite number of steps,
as the MU that stops offloading always has higher reluctance than that one that replaces it,
and the number of MUs is finite.
Let $\varoAP$ be the AP that the last MU that stopped offloading was connected to.
If the last MU that stopped offloading was replaced by an MU that did not offload before, then we reached a NE.
Otherwise some MUs may want to change to AP $\varoAP$. By Lemma~\ref{th:swapping-finite} if we only allow MUs to change between APs,
we terminate in a finite number of improvement steps. Now, no MU wants to stop offloading, and
no MU wants to start offloading either, because they did not want to do so before the MUs were allowed to change APs. Hence we reached a NE, which
proves the inductive step.

\end{proof}
\begin{figure}[t]
\ruleline{\emph{Update phase}}
\begin{algorithmic}[1]
\If{ $\varAP \in \vardeviatorsOtoL$} 
\State /* Corresponds to case (i) */
\State Let $\varplayer^\prime \gets (\vargamestage,1,\varAP)$
\State Let $\vardecisionsvector^\prime(\vargamestage)=(0,\vardecision_{-\varplayer^\prime}(\vargamestage))$/* Best reply by MU $\varplayer^\prime$ */
\ElsIf{ $\varoAP \in \vardeviatorsOtoL$ }
\State /* Corresponds to case (ii) */
\State Let $\vardecisionsvector^\prime(\vargamestage)=\vardecisionsvector(\vargamestage)$
\While{ $\vardeviatorsOtoL \not = \emptyset$ }
\State $\varoAP \gets \displaystyle{\argmax_{\varoAP^\prime \in \vardeviatorsOtoL}}\varreluctanceratio_{(\vargamestage,1,\varoAP^{\prime})}(\vardecisionsvector^\prime(\vargamestage))$\\\hfill/* AP with MU with highest reluctance */
\State Let $\varplayer^\prime \gets (\vargamestage,1,\varoAP)$
\State Let $\vardecisionsvector^{\prime}(\vargamestage) = (0,\vardecision^\prime_{-\varplayer^\prime}(\vargamestage))$\\\hfill /* Best reply by MU $(\vargamestage,1,\varoAP)$ */
\State $\vardeviatorsLtoO = \{\varplayer | \vardecision^{\prime}_{\varplayer}(\vargamestage)=0,
\varlocalcost_\varplayer \geq \varcloudcost_{\varplayer,\varoAP}(\vardecisionsvector^{\prime}(\vargamestage))\}$
\If{ $\vardeviatorsLtoO \not= \emptyset$ }
\State $\varplayer^\prime \gets \displaystyle{\argmax_{\varplayer \in \vardeviatorsLtoO}} \hspace{0.1cm} \varlocalcost_\varplayer$\\\hfill  /*MU with highest local cost*/
\State Let $\vardecisionsvector^\prime(\vargamestage)=(\varoAP,\vardecision^\prime_{-\varplayer^\prime}(\vargamestage))$\\\hfill /* Best reply by MU $\varplayer^\prime$ */
\State $\vardeviatorsOtoL\!\!=\!\!\{\varoAP \!\!\in\! \varAPs  |\exists \varplayer \!\in\! \varoffloaders_\varoAP(\vardecisionsvector^\prime(\vargamestage)),\varcloudcost_{\varplayer,\varoAP}(\vardecisionsvector^{\prime}(\vargamestage)) \!\!\geq\! \varlocalcost_\varplayer\}$
\Else
\State $\vardeviatorsOtoO = \{\varAP | \varAP \in \varAPs \setminus \{\varoAP\},\frac{\varnumberofusers_{\varoAP}(\vardecisionsvector^{\prime}(\vargamestage))+1}
{\varchannelbandwidth_{\varoAP}}\!\!<\!\!\frac{\varnumberofusers_\varAP(\vardecisionsvector^{\prime}(\vargamestage))}{\varchannelbandwidth_\varAP}\}$          
\While{$\vardeviatorsOtoO \not = \emptyset$}
\State $\varAP \gets \displaystyle{\argmax_{\varAP^{\prime} \in \vardeviatorsOtoO}}\varreluctanceratio_{(\vargamestage,1,\varAP^{\prime})}(\vardecisionsvector^\prime(\vargamestage))$\\\hfill /* AP with MU with highest reluctance */
\State Let $\varplayer^\prime \gets (\vargamestage,1,\varAP)$
\State Let $\vardecisionsvector^\prime(\vargamestage)=(\varoAP,\vardecision^\prime_{-\varplayer^\prime}(\vargamestage))$\\\hfill /* Best reply by MU $(\vargamestage,1,\varAP)$ */
\State Let $\varoAP \gets \varAP$
\State $\vardeviatorsOtoO = \{\varAP | \varAP \in \varAPs \setminus \{\varoAP\},\frac{\varnumberofusers_{\varoAP}(\vardecisionsvector^{\prime}(\vargamestage))+1}
{\varchannelbandwidth_{\varoAP}}\!\!<\!\!\frac{\varnumberofusers_\varAP(\vardecisionsvector^{\prime}(\vargamestage))}{\varchannelbandwidth_\varAP}\}$
\EndWhile
\EndIf
\EndWhile
\EndIf
\end{algorithmic}
\rule{\columnwidth}{0.5pt}
\caption{Pseudo code of the update phase of the distributed algorithm.}\label{fig:uphase}
\vspace{-0.4cm}
\end{figure}
As we next show, the above constructive proof provides a low complexity algorithm for computing a Nash equilibrium of the game. 
\begin{proposition}
  For the computation offloading game with non-elastic cloud, when player $\varplayerssetdim^{(\vargamestage)}$ enters the game in equilibrium $\vardecisionsvector^*(\vargamestage-1)$,
  a new Nash equilibrium can be computed in $O(\varplayerssetdim^{(\vargamestage)}+\varAPsdim)$ time.
\end{proposition}
\begin{proof}
Let us consider inductive step $\vargamestage$ in which MU $\varplayerssetdim^{(\vargamestage)}$ enters the game. From the proof of Theorem~\ref{theo:NE_ene}
it follows that if $\vardecision^*_{\varplayerssetdim^{(\vargamestage)}}(\vargamestage)=0$, or if $\vardecision^*_{\varplayerssetdim^{(\vargamestage)}}(\vargamestage)=\varAP\in\varAPs$
but none of the MUs want to deviate from their strategy in $\vardecisionsvector^*(\vargamestage-1)$, then a NE is reached without any update steps.
If $\vardecision^*_{\varplayerssetdim^{(\vargamestage)}}(\vargamestage)=\varAP\in\varAPs$ and case (i) happens, a NE is reached after one update step.
Now let us consider that $\vardecision^*_{\varplayerssetdim^{(\vargamestage)}}(\vargamestage)=\varAP\in\varAPs$ and case (ii) happens. Note that from Theorem~\ref{theo:NE_ene} 
it follows that case (ii) can happen only if $\varAPsdim>1$.
In what follows we characterize the longest sequences of update steps 
that lead to a NE for the case when $\varplayerssetdim^{(\vargamestage)}$ is even and when it is odd. 

If $\varplayerssetdim^{(\vargamestage)}$ is even, the worst case scenario is when 
$|\varoffloaders(\vardecisionsvector^{*}(\vargamestage-1))|=\lceil\frac{\varplayerssetdim^{(\vargamestage)}-1}{2}\rceil$ and $\varnumberofusers_\varAP(\vardecisionsvector^{*}(\vargamestage-1))=0$,
in the new strategy profile $\vardecisionsvector(\vargamestage)$ every MU $\varplayer\in\varoffloaders(\vardecisionsvector^{*}(\vargamestage-1))$ 
wants to change to local computing, and when MU $\varplayer$ with highest reluctance $\varreluctanceratio_\varplayer(\vardecisionsvector(\vargamestage))$ changes to local computing, all MUs 
$\varplayer\in\varplayersset\setminus\varoffloaders(\vardecisionsvector^{*}(\vargamestage-1))$, i.e., a total of $\lfloor\frac{\varplayerssetdim^{(\vargamestage)}-1}{2}\rfloor$ MUs 
would like to start offloading. In the corresponding  sequence of update steps that leads to a NE,  in the first $2\lfloor\frac{\varplayerssetdim^{(\vargamestage)}-1}{2}\rfloor+1$ 
update steps all MUs $\varplayer\in\varoffloaders(\vardecisionsvector^{*}(\vargamestage-1))$ stop to offload and all MUs $\varplayer\in\varplayersset\setminus\varoffloaders(\vardecisionsvector^{*}(\vargamestage-1))$ 
start to offload, and in the next $(\varAPsdim-1)$ update steps $(\varAPsdim-1)$ MUs change between APs.
Therefore, a NE is reached after at most $2\lfloor\frac{\varplayerssetdim^{(\vargamestage)}-1}{2}\rfloor+1+(\varAPsdim-1)$ update steps.

If $\varplayerssetdim^{(\vargamestage)}$ is odd, the worst case scenario is when $\varnumberofusers_\varAP(\vardecisionsvector^{*}(\vargamestage-1))=1$,
in the new strategy profile $\vardecisionsvector(\vargamestage)$ a total of 
$\lfloor\frac{\varplayerssetdim^{(\vargamestage)}-1}{2}\rfloor$ MUs of the $|\varoffloaders(\vardecisionsvector^{*}(\vargamestage-1))|=\lfloor\frac{\varplayerssetdim^{(\vargamestage)}-1}{2}\rfloor+1$ MUs
that offload want to change to local computing, and when MU $\varplayer$ with highest reluctance $\varreluctanceratio_\varplayer(\vardecisionsvector(\vargamestage))$ changes to local computing, all MUs 
$\varplayer\in\varplayersset\setminus\varoffloaders(\vardecisionsvector^{*}(\vargamestage-1))$, i.e., a total of $\lfloor\frac{\varplayerssetdim^{(\vargamestage)}-1}{2}\rfloor-1$ MUs 
would like to start offloading. Following the same reasoning as above, we obtain that a NE is reached after at most
$2(\lfloor\frac{\varplayerssetdim^{(\vargamestage)}-1}{2}\rfloor-1)+1+(\varAPsdim-1)$ update steps.

\end{proof}
Consider now that we add players one at a time, we then obtain the following bound on the complexity of computing a NE.
\begin{corollary}
A Nash equilibrium of the computation offloading game with non-elastic cloud can be computed in $O(\varplayerssetdim^2+\varplayerssetdim\varAPsdim)$ time.
\end{corollary}
So far we have shown that starting from a NE and adding a new player, a new NE can be computed. We now show a similar result for the case
when a player leaves.
\begin{theorem}
\label{theo:NE_depart}
Consider  the computation offloading game with  non-elastic cloud, and assume the system is in a NE. If an existing player leaves the game and the remaining
players play  best replies, they converge to a Nash equilibrium after a finite number of updates.
\end{theorem}
\begin{proof}
 Let us consider that player $\varplayer$ leaves the game, when the system is in a NE. 
 If the  strategy of player $\varplayer$ is to perform local computation, none of the remaining players would be affected when player $\varplayer$ leaves. 
 If the strategy of player $\varplayer$ is to offload using one of the APs, we can consider player $\varplayer$ as a player that after
 changing his strategy from offloading to local computing, would have no incentive to offload again. Recall from the
 proof of Theorem~\ref{theo:NE_ene} that when a player changes her strategy from offloading to local computing the game converges to a Nash equilibrium after a finite number of updates.
 This proves the theorem.  
\end{proof}
Observe that Theorem~\ref{theo:NE_ene} and Theorem \ref{theo:NE_depart} allow for efficient computation of equilibrium system operation if the time between user arrivals and departures
is sufficient to compute a new equilibrium. Furthermore, the computation can be done in a decentralized manner, by letting MUs perform best improvements one at a time. The advantage of such a decentralized
implementation could be that MUs do not have to reveal their parameters.

%% file: PoA.tex
\section{Price of Anarchy}
\label{sec::poa}
We have so far shown that NE exist and provided low complexity algortihms for computing a NE. We now address the important question
how far the system performance would be from optimal in a NE. To quantify the difference from the
optimal performance we use the price of anarchy (PoA), defined as the ratio of the worst case NE cost and the minimal
cost
\begin{equation}
  PoA =\frac{\max\limits_{\vardecisionsvector^{*}}\sum_{\varplayer \in \varplayersset}\varcostfunction_\varplayer(\vardecisionsvector^{*})}
 {\min\limits_{\vardecisionsvector\in\vardecisionsset}\sum_{\varplayer \in \varplayersset}\varcostfunction_\varplayer(\vardecisionsvector)}.
\end{equation}
In what follows we give an upper bound on the PoA.
\begin{theorem}
  The price of anarchy for the computation offloading game is upper bounded by
$$\frac{\sum_{\varplayer \in \varplayersset}\varlocalcost_\varplayer}{\sum_{\varplayer\in\varplayersset}\min\{\varlocalcost_\varplayer,\bar{\varcloudcost_{\varplayer,1}},...,
 \bar{\varcloudcost_{\varplayer,\varAPsdim}}\}},$$ 
both in the case of elastic cloud and in the case of non-elastic cloud. 
\end{theorem}
\begin{proof}
  First we show that if there is a NE in which all players perform local computation then it is the worst case NE.
 To show this let $\vardecisionsvector^{*}$ be an arbitrary NE. 
 Observe that $\varcostfunction_{\varplayer}(\vardecision_{\varplayer}^{*},\vardecision_{-\varplayer}^{*})
 \leq \varlocalcost_\varplayer$ holds for every player $\varplayer \in \varplayersset$. Otherwise, if $\exists \varplayer \in \varplayersset$ such that
 $\varcostfunction_{\varplayer}(\vardecision_{\varplayer}^{*},\vardecision_{-\varplayer}^{*}) > \varlocalcost_\varplayer$, player $\varplayer$ would have an incentive to deviate from 
 decision $\vardecision_{\varplayer}^{*}$, which contradicts our initial assumption that $\vardecisionsvector^{*}$ is a NE. Thus in any NE 
 $\sum_{\varplayer\in\varplayersset}\varcostfunction_\varplayer(\vardecision_\varplayer^{*},\vardecision_{-\varplayer}^{*})\leq\sum_{\varplayer\in\varplayersset}\varlocalcost_\varplayer$
 holds, and if all players performing  local computation is a NE then it is the worst case NE.\\
 Now we derive a lower bound for the optimal solution of the computation offloading game in the case of both the elastic and non-elastic cloud. Let us consider  an arbitrary decision profile 
 $(\vardecision_\varplayer,\vardecision_{-\varplayer})\in\vardecisionsset$. If $\vardecision_\varplayer=0$, then $\varcostfunction_\varplayer(\vardecision_\varplayer,\vardecision_{-\varplayer})=\varlocalcost_\varplayer$. 
 Otherwise, if $\vardecision_\varplayer=\varAP$ for some $\varAP\in\varAPs$, we have that in the best case $\vardecision_\varoplayer=0$ for every $\varoplayer\in\varplayersset\setminus\{\varplayer\}$,
and thus $\varnumberofusers(\vardecisionsvector)=1$. Therefore, $\varuplinkrate_{\varplayer}^{\varAP}(\vardecision_\varplayer,\vardecision_{-\varplayer})\leq \varchannelbandwidth_\varAP$ and 
 $\varcloudcapability_\varplayer\leq\varcloudcapability$, which implies that
 \begin{align}
  \nonumber & \varcloudcost_{\varplayer,\varAP}(\vardecision_\varplayer,\vardecision_{-\varplayer}) = 
  (\vartimeweight_\varplayer+\varenergyweight_\varplayer\varpower_\varplayer)\frac{\vardatasize_\varplayer}{\varuplinkrate_{\varplayer}^{\varAP}(\vardecision_\varplayer,\vardecision_{-\varplayer})} + 
  \vartimeweight_\varplayer\frac{\varCPUcyclesnumber_\varplayer}{\varcloudcapability_\varplayer}\\ \nonumber
  & \geq (\vartimeweight_\varplayer+\varenergyweight_\varplayer\varpower_\varplayer)\frac{\vardatasize_\varplayer}{\varchannelbandwidth_\varAP} 
  + \vartimeweight_\varplayer\frac{\varCPUcyclesnumber_\varplayer}{\varcloudcapability}=\bar{\varcloudcost_{\varplayer,\varAP}}.\nonumber
 \end{align}
Hence, we have   $\varcostfunction_\varplayer(\vardecision_\varplayer,\vardecision_{-\varplayer})\!\geq\!\min\{\varlocalcost_\varplayer,\bar{\varcloudcost_{\varplayer,1}},...,\bar{\varcloudcost_{\varplayer,\varAPsdim}}\}$
and $\sum_{\varplayer\in\varplayersset}\varcostfunction_\varplayer(\vardecision_\varplayer,\vardecision_{-\varplayer})
\!\geq\! \sum_{\varplayer \in \varplayersset} \min\{\varlocalcost_\varplayer,\bar{\varcloudcost_{\varplayer,1}},...,\bar{\varcloudcost_{\varplayer,\varAPsdim}}\}$. Using these we can establish the following bound
\begin{align}
 \nonumber PoA\!=\!\frac{\max\limits_{\vardecisionsvector^{*}}\!\sum_{\varplayer \in \varplayersset}\!\varcostfunction_\varplayer(\vardecisionsvector^{*})}
 {\min\limits_{\vardecisionsvector\in\vardecisionsset}\!\sum_{\varplayer \in \varplayersset}\!\varcostfunction_\varplayer(\vardecisionsvector)}
 \!\leq\!\frac{\!\sum_{\varplayer \in \varplayersset}\!\varlocalcost_\varplayer}{\!\sum_{\varplayer\in\varplayersset}\!\min\{\varlocalcost_\varplayer,\!\bar{\varcloudcost_{\varplayer,1}}\!,\!...,
 \!\bar{\varcloudcost_{\varplayer,\varAPsdim}}\}},
\end{align}
which proves the theorem.
\end{proof}

%% file: numerical.tex
\section{Numerical Results}
\label{sec::numerical}
\begin{figure*}[tb]
\begin{minipage}{0.485\textwidth}
\vspace{-0.2cm}
\begin{center}
  \includegraphics[width=\columnwidth]{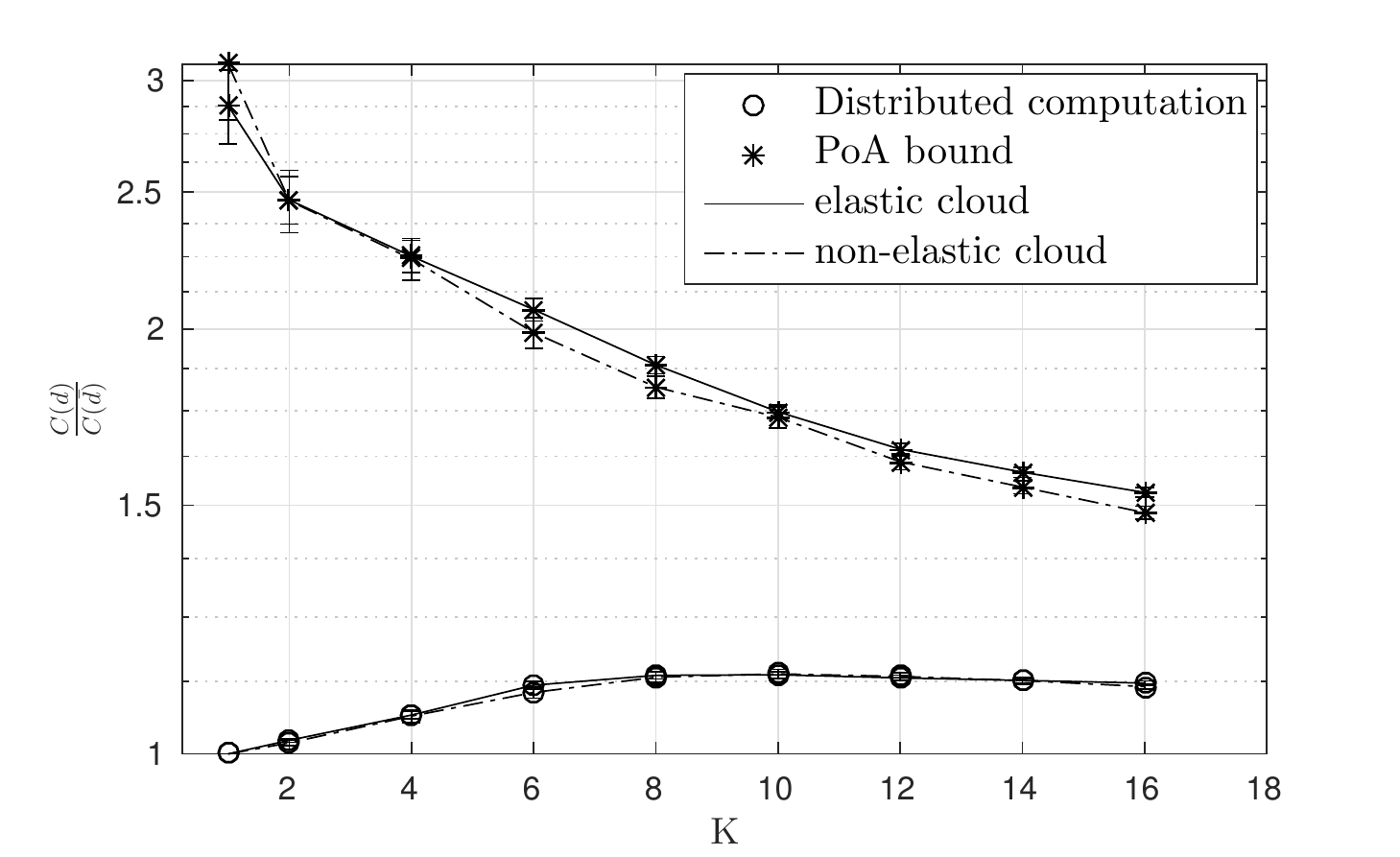}
  \caption{Ratio between the total cost achieved by the proposed distributed algorithm and the optimal total cost for the elastic and non-elastic cloud,
  $\varAPsdim = 3$. The results shown are the averages of 500 simulations, together with 95\% confidence intervals.}
  \label{fig:total_cost}
\end{center}
\vspace{-0.8cm}
\end{minipage}
\hspace{0.015\textwidth}
\begin{minipage}{0.485\textwidth}
\vspace{-0.4cm}
 \begin{center}
  \includegraphics[width=\columnwidth]{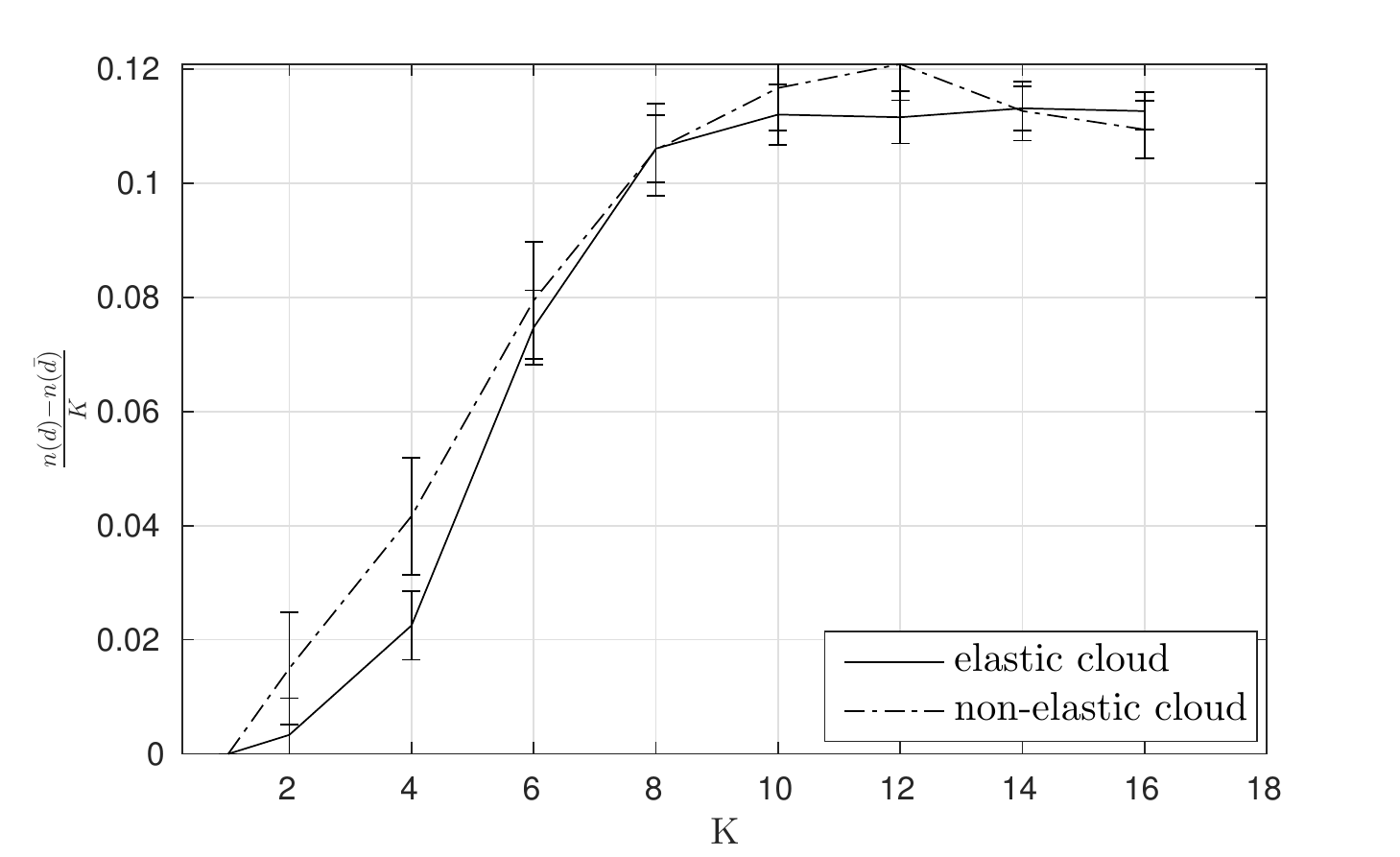}
  \caption{Offloading ratio vs. number of users $\varplayerssetdim$ for the elastic and non-elastic cloud,
  $\varAPsdim=3$. The results shown are the averages of 500 simulations, together with 95\% confidence intervals.}
  \label{fig:num_offlo}
  \end{center}
\vspace{-0.8cm}
\end{minipage}
\end{figure*}
\begin{figure}[tb]
\vspace{-0.2cm}
 \begin{center}
  \includegraphics[width=\columnwidth]{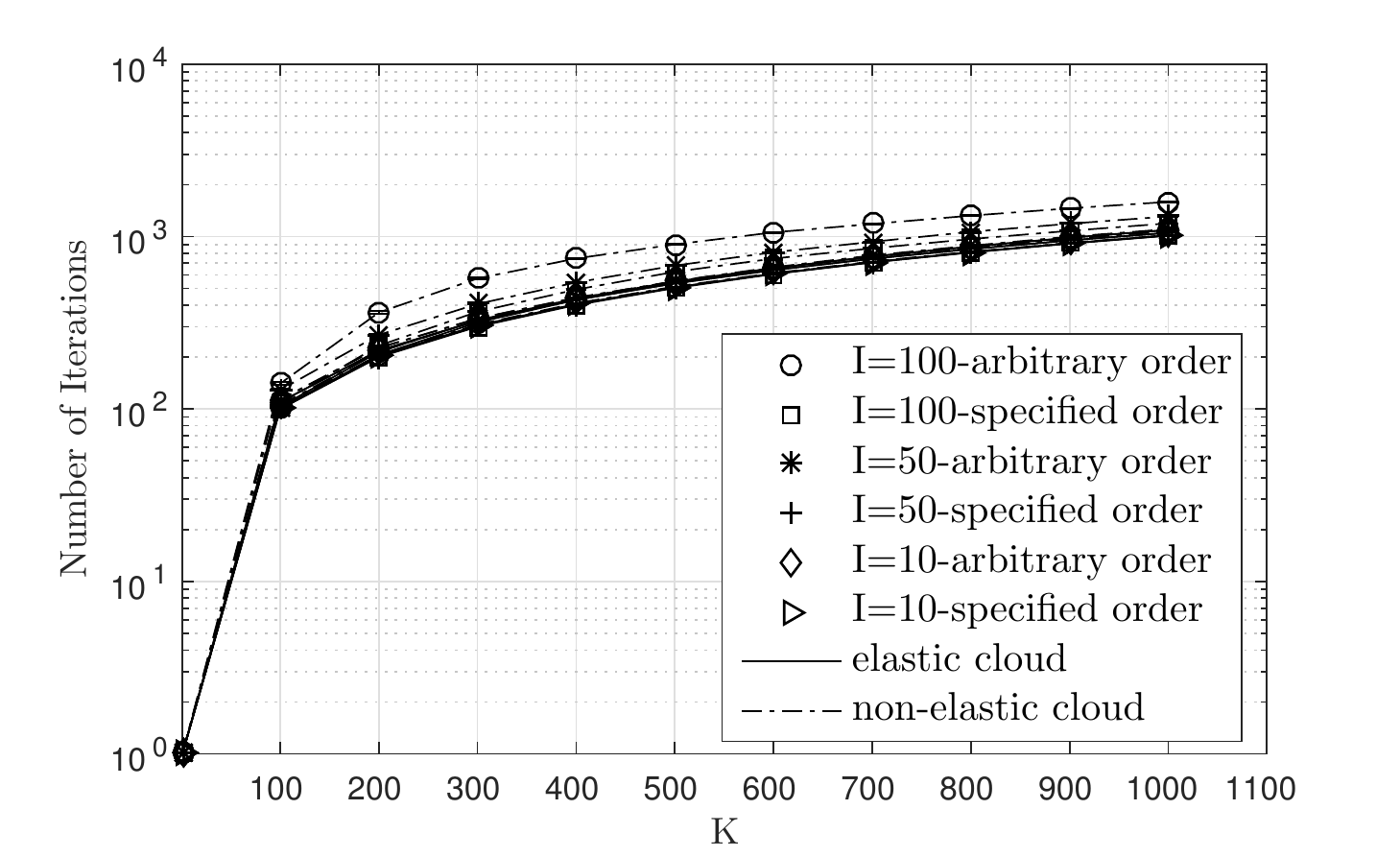}
  \caption{Number of iterations vs. number of users $\varplayerssetdim$ for the elastic and non-elastic cloud,
  $\varAPsdim = $10, 50 and 100. The results shown are the averages of 100 simulations, together with 95\% confidence intervals.}
  \label{fig:num_iter}
  \end{center}
 \vspace{-0.9cm}
\end{figure}
We use simulations to evaluate the cost performance and the computational time of the proposed distributed algorithms.
\subsection{Evaluation Scenario}
In all configurations, we consider that the bandwidth of each AP is drawn from a normal distribution with mean $\mu=5$ \textit{MHz} and standard deviation of $0.2\mu$. The parameters 
 $<\!\!\vardatasize_\varplayer,\varCPUcyclesnumber_\varplayer\!\!>$ that characterize the computation tasks, the computational capability
of MU $\varusercapability_\varplayer$ and the weights attributed to energy consumption $\varenergyweight_\varplayer$ and the time it takes to finish the computation 
$\vartimeweight_\varplayer$ were drawn 
from a continuous  uniform distribution with parameters $[0.42, 2]$Mb, $[0.1,0.8]$~\textit{Gigacycles}, $[0.5, 1]$~\textit{Gigacycles}, $[0,1]$ and $[0,1]$, respectively.
The consumed energy per CPU cycle $\varenergyconstant_\varplayer$ was set to $10^{-11}(\varusercapability_{\varplayer})^{2}$ according to
measurements reported in~\cite{Wen12Infocom,miettinenenergy}. The data transmit power $\varpower_\varplayer$ was set to 0.4\textit{W} according to~\cite{balasubramanian2009energy}. 
In the case of the non-elastic cloud, the computation capability of the cloud $\varcloudcapability$ was set to $100$~\textit{Gigacycles}~\cite{soyata2012cloud} and in the case of the elastic cloud each MU that
offloads receives $\varcloudcapability_\varplayer= 100$~\textit{Gigacycles} amount of computing power.

In order to evaluate the cost performance of the equilibrium strategy profile $\vardecisionsvector^*$ computed by the proposed distributed algorithms, 
we computed the optimal strategy profile $\bar{\vardecisionsvector}$ that minimizes the total cost, i.e.,
$\bar{\vardecisionsvector}=\arg\min_\vardecisionsvector\sum_{\varplayer\in\varplayersset}{\varcostfunction_{\varplayer}(\vardecisionsvector)}$.
Furthermore, as a baseline for comparison we use the system cost that can be achieved in the strategy profile
in which all MUs execute their computation tasks locally, which coincides with the bound on the PoA.

\subsection{Price of Anarachy}
Figure~\ref{fig:total_cost} shows the cost ratio $\varcostfunction(\vardecisionsvector^*)/\varcostfunction(\bar{\vardecisionsvector})$ in the case of the
elastic as well as in the case of the non-elastic cloud. To make the computation of the optimal strategy profile $\bar{\vardecisionsvector}$ feasible, we considered
a scenario with $\varAPsdim=3$ APs and we show the cost ratio $\varcostfunction(\vardecisionsvector^{*})/\varcostfunction(\bar{\vardecisionsvector})$ as a function of the number of MUs. 
Figure~\ref{fig:total_cost} shows that the results reached by the algorithms are close to the optimal results and the difference between the elastic cloud case and
the non-elastic cloud case is negligible. Furthermore, we can observe that the cost ratio increases slightly up to $\varplayerssetdim=6$ MUs and from that point it remains fairly unchanged.
This is due to the number of MUs that choose to offload, as we will see later.
The upper bound on the PoA, which is also shown in Figure~\ref{fig:total_cost}, additionally confirms that the proposed distributed algorithms perform good in terms of the cost ratio. It is interesting to note that the gap between the PoA bound and the actual cost ratio decreases with increasing number of MUs.

To get insight into the structure of the equilibrium strategy profile $\vardecisionsvector^{*}$ computed by the distributed algorithms,
we compare the number of MUs that offload in equilibrium and the number of MUs that offload in the optimal strategy profile $\bar{\vardecisionsvector}$,
by computing the offloading difference ratio $(\varnumberofusers(\vardecisionsvector^{*})-\varnumberofusers(\bar{\vardecisionsvector}))/\varplayerssetdim$. 
Figure~\ref{fig:num_offlo} shows the offloading difference ratio
corresponding to the results shown in Figure~\ref{fig:total_cost}. The results show that the offloading difference ratio is fairly small in the case of the
elastic cloud as well as in the case of the non-elastic cloud. As the number of MUs increases, the offloading difference ratio increases too, which explains the
increased cost ratio observed in Figure~\ref{fig:total_cost}, as more offloaders reduce the achievable rate, which in turn leads to increased costs.
The observation that the number of MUs that offload is higher in equilibrium than in the optimal solution
is consistent with the theory of the tragedy of the commons in economic theory~\cite{Hardin1968Science}.

\subsection{Computational Complexity}
In order to evaluate the computational complexity of the proposed algorithms, we study the number of iterations needed
to compute the strategy profile $\vardecisionsvector^{*}$ for three  scenarios with $\varAPsdim=10,50,100$ APs, respectively. 
For the elastic cloud the number of iterations is the number of update steps, while for the non-elastic cloud the number of iterations is the sum of the update steps over all induction steps.
Figure~\ref{fig:num_iter} shows the number of iterations as a function of the number of MUs. For the non-elastic cloud we consider two orderings of adding MUs: in the first case the MUs are added in random order, while in the second case the  MUs enter the game in increasing order of their ratio $\frac{\vardatasize_{\varplayer}}{\varlocalcost_{\varplayer}\varCPUcyclesnumber_{\varplayer}}$.
In both cases we use the same simulation scenarios in order to compute the number of the iterations. Intuitively, we can expect 
that the second case results in a smaller number of the iterations, since the MUs with lower $\frac{\vardatasize_{\varplayer}}{\varlocalcost_{\varplayer}\varCPUcyclesnumber_{\varplayer}}$ ratio 
have higher computational capability to execute computationally more demanding tasks with smaller offloading data size than the MUs with higher $\frac{\vardatasize_{\varplayer}}{\varlocalcost_{\varplayer}\varCPUcyclesnumber_{\varplayer}}$ ratio.
However, the simulation results show that the number of iterations is fairly insensitive to the order of adding the MUs and mostly depends on the number of MUs $\varplayerssetdim$. This insensitivity allows for the implementation of a very low-overhead decentralized solution, as the coordinator need not care about the order in which the MUs are added for computing the equilibrium allocation.

%% file: related.tex
\section{Related Work}
\label{sec::related}

There is a significant body of works that deals with the design of energy efficient computation
offloading for a single mobile user~\cite{Cuervo2010MobiSys,Wen12Infocom,Barbera2013Infocom,Kumar2013MobNet,Rudenko:1998:SPC:584007.584008,Huang2012twc,Hyytia2015wowmom}.
The experimental results in~\cite{Rudenko:1998:SPC:584007.584008} showed that significant battery power savings can be achieved
by computation offloading.~\cite{Barbera2013Infocom} studied the commmunication overhead of computation offloading and the impact of bandwidth availability on an experimental platform.
~\cite{Cuervo2010MobiSys} proposed a code partitioning solution for fine-grained energy-aware computation offloading.
~\cite{Huang2012twc} proposed an algorithm for offloading partitioned code under bandwidth and delay constraints.
~\cite{Wen12Infocom} proposed CPU frequency and transmission power adaptation for energy-optimal computation offloading under delay constraints.
~\cite{Hyytia2015wowmom} modeled the offloading problem under stochastic task arrivals as a Markov decision process and provided a near-optimal offloading policy.

A number of recent works considered the problem of joint energy minimization for multiple mobile users~\cite{Yang:2013:FPE:2479942.2479946,Rahimi2013Cloud,Sardellitti2015tsipn}. 
~\cite{Yang:2013:FPE:2479942.2479946} studies computation partitioning for streaming data processing with the aim of maximizing throughput, considering sharing of 
computation instances among multiple mobile users, and  proposes a genetic algorithm as a heuristic for solving the resulting optimization problem.
~\cite{Rahimi2013Cloud} models computation offloading to a tiered cloud infrastructure under user mobility in a location-time workflow framework, and proposes a heuristic for minimizing the
users' cost.
~\cite{Sardellitti2015tsipn} aims at minimizing the mobile users' energy consumption by joint allocation of radio resources and cloud computing power, and provides an iterative
algorithm to find a local minimum of the optimization problem. 

A few recent works provided a game theoretic treatment of computation offloading in a game theoretical setting~\cite{Wang2013Sose,Cardellini2015,Chen2015tpds,Chen2015ToN,Meskar2015ICC,Ma:2015:GAC:2811587.2811598}.
~\cite{Wang2013Sose} considers a two-stage problem, where first each mobile user decides what share of its task to offload so as to minimize its energy consumption and to meet its delay deadline,
and then the cloud allocates computational resources to the offloaded tasks.
~\cite{Cardellini2015} considers a two-tier cloud infrastructure and stochastic task arrivals and proves the existence of equilibria and provides an algorithm for computing and equilibrium.
~\cite{Meskar2015ICC} considers tasks that arrive simultaneously, a single wireless link, and elastic cloud, and show the existence of equilibria when all mobile users have the same delay budget.
Our work differs from~\cite{Wang2013Sose} in that we consider that the allocation of cloud resources is known to the mobile users, from~\cite{Cardellini2015} in that we take into account
contention in the wireless access, and from~\cite{Meskar2015ICC} in that we consider multiple wireless links and a non-elastic cloud.

Most related to our work are the problems considered in~\cite{Chen2015tpds,Chen2015ToN,Ma:2015:GAC:2811587.2811598}.~\cite{Chen2015tpds} considers contention on a single wireless link and an elastic cloud, assumes upload rates to be determined by the Shannon capacity of an interference channel, and shows that the game is a potential game.
~\cite{Chen2015ToN} extends the model to multiple wireless links and shows that the game is still a potential game. Unlike these works, we consider fair bandwidth sharing and consider the case of non-elastic cloud.
~\cite{Ma:2015:GAC:2811587.2811598} considers multiple wireless links, fair bandwidth sharing and a non-elastic cloud, and claims the game to have an exact potential.
In our work we on the one hand extend the model to an elastic cloud, on the other hand we show that an exact potential cannot exist in case of a non-elastic cloud, but at the same time we prove the existence of an equilibrium allocation, provide an efficient algorithm with quadratic complexity for computing one, and provide a bound on the price of anarchy.

Besides providing efficient algorithms for computing equilibria, the importance of our contribution lies in the fact that while games with an elastic cloud are player-specific singleton congestion games for which the existence of equilibria is known~\cite{Milchtaich1996111}, the non-elastic cloud model does not fall in this category of games and thus no general equilibrium existence result exists.

%% file: conclusion.tex
\section{Conclusion}
\label{sec::conclusion}

We have considered the problem of computation offloading in a multi-access wireless network
by self-interested mobile users for mobile cloud computing, for the case of elastic and non-elastic cloud resources.
We provided a game theoretical formulation of the problem, and showed that in the case of
an elastic cloud a simple algorithm, in which users iteratively improve their allocations,
can be used for computing an equilibrium. We showed that the same algorithm may fail in the case of a
non-elastic cloud, but also showed that an equilibrium always exists, and provided an algorithm for computing an equilibrium
with quadratic complexity. Finally, we provided a bound on the price of anarchy. Simulation results show
that the complexity bound is not tight, and the proposed algorithm scales better than quadratic in terms of the
number of users, and the obtained equilibria provide good system performance.